\documentclass[10pt,journal]{IEEEtran}

\usepackage{algorithm}
\usepackage{algorithmic}
\usepackage{amsmath}
\usepackage{amsthm}
\usepackage{amssymb}
\usepackage{multicol}
\usepackage{multirow}
\usepackage{graphicx}
\usepackage[mathscr]{euscript}
\usepackage{multirow}
\usepackage{centernot}
\usepackage{listliketab}
\usepackage{cutwin}
\usepackage{verbatim}
\usepackage{slashbox}
\usepackage{color}
 \usepackage{multirow}
 \usepackage{wrapfig}
 \usepackage{bm}
 \usepackage[hidelinks]{hyperref}

\usepackage{array}
\newcolumntype{L}[1]{>{\raggedright\let\newline\\\arraybackslash\hspace{0pt}}m{#1}}
\newcolumntype{C}[1]{>{\centering\let\newline\\\arraybackslash\hspace{0pt}}m{#1}}
\newcolumntype{R}[1]{>{\raggedleft\let\newline\\\arraybackslash\hspace{0pt}}m{#1}}

\newtheorem{observation}{Observation}
\newtheorem{example}{Example}
\newtheorem{theorem}{Theorem}
\newtheorem{lemma}{Lemma}

\newtheorem{remark}{Remark}

\newtheorem{property}{Property}

\newcommand\T{\rule{0pt}{2.1ex}}
\newcommand\B{\rule[-.7ex]{5pt}{0pt}}

\begin{document}

\title{Information Diffusion in Social Networks in Two Phases
}
\markboth{IEEE Transactions on Network Science and Engineering, volume 3, number 4, pages 197-210, 2016}{Swapnil Dhamal, Prabuchandran K.J., and Y. Narahari: Information Diffusion in Social Networks in Two Phases}
\author{Swapnil Dhamal, Prabuchandran K.J., and Y. Narahari 
\thanks{
Cite this article as: 
S. Dhamal, K. J. Prabuchandran, and Y. Narahari, ``Information
diffusion in social networks in two phases,'' IEEE Transactions on
Network Science and Engineering, vol. 3, no. 4, pp. 197-210, 2016.
The original publication is available at \url{http://ieeexplore.ieee.org/abstract/document/7570252/}
%
}
}

\maketitle

\begin{abstract}
The problem of maximizing information diffusion, given a certain budget expressed in terms of the number of seed nodes, is an important topic in social networks research. Existing literature focuses on single phase diffusion where all seed nodes are selected at the beginning of diffusion and all the selected nodes are activated simultaneously. This paper undertakes a detailed investigation of the effect of selecting and activating seed nodes in multiple phases. Specifically, we study diffusion in two phases assuming the well-studied independent cascade model. First, we formulate an objective function for two-phase diffusion, investigate its properties, and propose efficient algorithms for finding seed nodes in the two phases. Next, we study two associated problems: (1) {\em budget splitting} which seeks to optimally split the total budget between the two phases and (2) {\em scheduling} which seeks to determine an optimal delay after which to commence the second phase. Our main conclusions include: (a) under strict temporal constraints, use single phase diffusion, (b) under moderate temporal constraints, use two-phase diffusion with a short delay while allocating most of the budget to the first phase, and (c) when there are no temporal constraints, use two-phase diffusion with a long delay while allocating roughly one-third of the budget to the first phase.
\end{abstract}


\section{Introduction}
\label{sec:intro_mpid}
Social networks play a fundamental role in the spread of information on a large scale. 
An information can be of various types: opinions, behaviors, innovations, diseases, rumors, etc. 
Depending on whether we aim to maximize or restrict the spread of information, the objective function can be defined accordingly.
One of the central questions in information diffusion is: given a certain budget $k$ expressed in terms of the number of seed nodes, which $k$ nodes in the social network should be selected to trigger the diffusion so as to maximize a suitably defined objective function?

In this paper, we focus on the problem of {\em influence maximization\/}, where the objective function is the extent of information or influence spread.
For example, if a company wishes to do viral marketing of its product,
its objective would be to spread the information through the network so that it reaches large number of potential customers. So the company would try to select the seed nodes (nodes to whom free samples, discounts, or other such incentives are to be provided) such that the number of influenced nodes and hence the product sales, would be maximized.

\subsection{Model for Information Diffusion}

We represent social network as a weighted and directed graph $G = (N, E, \mathcal{P})$, where $N$ is the set of $n$ nodes, $E$ is the set of $m$ directed edges, and $\mathcal{P}$ is the set of weights associated with the edges. 
For studying diffusion in such a network, several models  have been proposed in the literature \cite{networkscrowdsmarkets}. 
The Independent Cascade (IC) model and the Linear Threshold (LT) model are two of the most well-studied models.
In this paper, our focus will be on the IC model; we later provide a note on the LT model.

{\scriptsize{$\bullet$}} 
{\em The Independent Cascade (IC) model\/}:
In this model, for each directed edge $(u,v) \in E$, there is an associated weight or {\em influence probability\/} $p_{uv}$ that specifies the probability with which source node $u$  influences target node $v$. The diffusion starts at time step $0$ with simultaneous triggering of a set of initially activated or influenced seed nodes, following which, it proceeds in discrete time steps. In each time step, nodes which got influenced in the previous time step (call them {\em recently activated nodes}) attempt to influence their neighbors, and succeed with the influence probabilities associated with the edges. 
These neighbors, if successfully influenced, now become recently activated nodes for the next time step. In any given time step, only recently  activated nodes contribute to diffusing information. After this time step, such nodes are no longer recently activated (call them {\em already activated nodes}). Nodes, once activated, remain activated for the rest of the diffusion. 
In short, when node $u$ gets activated at a certain time step, it gets exactly one chance to activate each of its inactive neighbors (that too in the immediately following time step), with 
probability $p_{uv}$ for each neighbor $v$.
The diffusion terminates when no further nodes can be activated.

{\scriptsize{$\bullet$}} 
{\em Notion of Live Graph\/}:
The notion of {\em live graph} is crucial to the analysis of the IC model.
A live graph $X$ is an instance of graph $G$, obtained by sampling the edges;
an edge $(u,v)$ is present in the live graph with probability $p_{uv}$ and absent with probability $1-p_{uv}$, independent of the presence of other edges in the live graph (so a live graph is directed and unweighted).
The probability $p(X)$ of occurrence of any live graph $X$, can be obtained as
$\prod_{(u,v) \in X} (p_{uv}) \prod_{(u,v) \notin X} (1-p_{uv})$.
It can be seen that as long as a node $u$, when influenced, in turn influences node $v$ with probability $p_{uv}$ that is independent of time, sampling the edge $(u,v)$ in the beginning of the diffusion is equivalent to sampling it when $u$ is activated \cite{kempe2003maximizing}. 

{\scriptsize{$\bullet$}} 
{\em Special Cases of the IC model\/}:
In this paper, when there is a need for transforming an undirected unweighted network (dataset) into a directed weighted network for studying the diffusion, we consider two popular, well-accepted special cases of the IC model, namely, the {\em weighted cascade (WC) model} and the {\em trivalency (TV) model}. The WC model does the transformation by making all edges bidirectional and assigning a weight to every directed edge $(u,v)$ equal to the reciprocal of $v$'s degree in the undirected network~\cite{kempe2003maximizing}.
The TV model makes all edges bidirectional and assigns a weight to every directed edge by uniformly sampling from the set of values $\{0.001, 0.01, 0.1\}$.

\subsection{Relevant Properties of Set Functions}

In the considered problem, we need to select a set of seed nodes based on its value (the extent of influence spread) which can be given by a set function. 
A {\em set function} $f(\cdot)$ is a function that takes a subset of $N$ as input and outputs a real number, that is, $f : 2^{N} \rightarrow \mathbb{R}$ where $2^{N}$ is the power set of $N$. 
 %
$f(\cdot)$ is said to be:

{\scriptsize{$\bullet$}} 
 {\em Non-negative} if 
$
f(S) \geq 0, \, \forall S \subseteq N
$.
%
%

{\scriptsize{$\bullet$}} 
{\em Monotone increasing} if 
$
f(S) \leq f(T), \, \forall S \subset T \subseteq N
$.
%
%

{\scriptsize{$\bullet$}} 
{\em Submodular} if 
$
f(S \cup \{i\}) - f(S) \geq f(T \cup \{i\}) - f(T), \, \forall i \in N \setminus T, \, \forall S \subset T \subset N
$,
%
that is, the marginal value added by a node to a superset of a set is not more than the marginal value added by that node to that set (diminishing returns property). 
It is said to be {\em supermodular} if the inequality is reversed.

{\scriptsize{$\bullet$}} 
{\em Subadditive} if 
$
f(S \cup T) \leq f(S) + f(T), \, \forall S , T \subseteq N
$,
that is, the value of a union of any two sets is at most the sum of their individual values.
It is said to be {\em superadditive} if the inequality is reversed.
It can be shown that a non-negative submodular function is subadditive, while a non-negative superadditive function is supermodular.

These properties have implications on which algorithms are likely to find a set with a good function value.
%
%
%
For instance, 
the greedy hill-climbing algorithm (selecting nodes one at a time, each time choosing a node that provides the largest marginal increase in the function value, until the budget is exhausted) provides an approximation guarantee of $(1- \frac{1}{e})$ for maximizing a non-negative, monotone increasing, submodular function~\cite{nemhauser1978analysis}.
Also there exists an algorithm that provides an approximation guarantee of $\frac{1}{2}$ for maximizing a subadditive function
\cite{feige2009maximizing}. 

\subsection{Relevant work}
\label{sec:relevant_mpid}

The problem of influence maximization in social networks has been extensively studied in the literature \cite{networkscrowdsmarkets,guille2013information}.
Chen, Wang, and Yang \cite{chen2010scalable} show that obtaining the exact value of the objective function for a seed set (the expected number of influenced nodes at the end of the diffusion that was triggered at the nodes of that set), under the IC model, is \#P-hard. 
They show that the value can be obtained with high accuracy using a sufficiently large number of Monte-Carlo simulations.
Kempe, Kleinberg, and Tardos \cite{kempe2003maximizing} show that maximizing the objective function under the IC model is NP-hard,
and present a $(1-\frac{1}{e}-\epsilon)$-approximate algorithm,
where $\epsilon$ is small for sufficiently large number of Monte-Carlo simulations.
Chen, Wang, and Yang \cite{chen2009efficient} propose fast heuristics for influence maximization in the IC model.
%
There have been
attempts to relax the assumption that influence probabilities are known
\cite{goyal2010learning}.

Narayanam and Narahari \cite{narayanam2010shapley} provide an algorithm 
that gives satisfactory performance irrespective of whether or not the objective function is submodular.
%
Franks et al. \cite{franks2013manipulating} use influencer agents effectively to manipulate the emergence of conventions and increase convention adoption and quality.
Shakarian et al. \cite{shakarian2013mancalog} introduce a logical framework designed to describe cascades in complex networks.
%

Another well-studied problem is the problem of influence limitation in social networks \cite{budak2011limiting, premm2012influence}, where the objective is to minimize the spread of a negative campaign by triggering a positive campaign.

Time related constraints in the context of diffusion have also been studied in the literature. 
Chen, Lu, and Zhang \cite{chen2012time} consider the problem where the goal is to maximize influence spread within a given deadline.
Nguyen et al. \cite{nguyen2012containment} aim to find the smallest set of influential nodes whose decontamination with good information would help contain the viral spread of misinformation, that was initiated from a given set, to a desired ratio in a given number of time steps.

The above papers address only single phase diffusion.
The idea of using multiple phases for maximizing an objective function has been presented in \cite{golovin2011adaptive}; the study is a preliminary one.
To the best of our knowledge, ours is the first detailed effort to study multi-phase information diffusion in social networks.
A previous, very preliminary, concise version of this paper appears in \cite{dhamal2015multiphase}.
In the next section, we bring out the motivation for this work, present a motivating example, and
describe the agenda of this work.

\section{Motivation and Agenda}
\label{sec:motiv_mpid}

Most of the existing literature on information diffusion works with the assumption that the diffusion is triggered at all the $k$ seed nodes in one go, that is, the budget is exhausted in one single instalment. We consider triggering the diffusion in multiple phases by appropriately splitting 
the budget $k$ across the phases. 

In the IC model, 
the diffusion is a random process. Since the general problem addressed in the literature aims to maximize influence spread in expectation, it is possible that the spread in certain instances is much less than the expected one. This is a vital practical issue because a company or organization investing in selecting seed nodes cannot afford awkward instances where the spread is disappointingly low. Multi-phase diffusion seems an attractive and  natural approach wherein, the company can modulate its decisions at intermediate times during the diffusion process, in order to avoid such instances. This happens because the company would be more certain about the diffusion and hence would hopefully select better seed nodes in the second and subsequent phases.  However, there is a delay in activating the second and subsequent seed sets and the overall diffusion process may be delayed, leading to compromise of time. 
This may be undesirable when the value of the product or information decreases with time, or when there is a competing diffusion and people get influenced by the product or information which reaches them first.

There is thus a natural trade-off between (a) using better knowledge of influence spread to increase the number of influenced nodes at the end of the diffusion and (b) the accompanying delay in the activation of seed sets from the second phase onwards.

For multi-phase diffusion to be implemented effectively, it is necessary that the company is able to observe the status of nodes in the social network (inactive, recently activated, or already activated), that is, the company needs to link its customers to the corresponding nodes in the social network. To make such an observation, it would be useful to get the online social networking identity (say Facebook ID) of a customer as soon as the customer buys the product. This could be done using a product registration website (say for activating warranty) where a customer, on buying the product, needs to login using a popular social networking website (say Facebook), or needs to provide an email address that can be linked to a Facebook ID. Thus the time step when the node buys the product, can be obtained, and hence the node can be classified as already activated or recently activated.

In this paper, to obtain a firm grounding on multi-phase diffusion, we focus our attention on two-phase diffusion. We believe that much of the intuition from this work carries over to multi-phase diffusion.

Note that the start of the second phase does not kill the diffusion that commenced in the first phase. When the second phase commences, the recently activated nodes (due to the first phase) effectively act as seed nodes for the second phase (in addition to the seed nodes that are separately selected for second phase).

\subsection{A Motivating Example}
\label{sec:example}
We now illustrate two-phase diffusion  with a simple stylized example.
Consider the graph in Figure~\ref{fig:motiv_mpid}(a) where the influence probabilities are as shown. Activation of node $A$ or $B$ or $C$ results in activation of 100 additional nodes each, in the following time step, with probability 1.
Consider a total budget of $k=2$. 
Assume that the live graph in Figure~\ref{fig:motiv_mpid}(b) is destined to occur (we do not have this information at time step 0).
Consider a typical influence maximization algorithm.

Let us study single-phase diffusion on this graph. Let $A$ and $B$ be the two seed nodes selected by 
the algorithm in time step 0. 
In time step 1 as per IC model, 200 additional nodes get influenced.
Since the realized live graph is as shown in Figure~\ref{fig:motiv_mpid}(b), the diffusion stops as there is no outgoing edge from the recently activated nodes to any inactive node. So the diffusion stops at time step \textbf{1}, with \textbf{202} influenced~nodes.

For two-phase diffusion, let the total budget of $2$ be split as $1$ each for the two phases,
and let the second phase be scheduled to start in time step 3. 
%
Now let us say that at time step 0, the algorithm selects $A$ as the only seed node for first phase. In time step 1, it influences its set of 100 nodes and also node $B$.
In time step 2, 
$B$'s set of 100 nodes get influenced. But more importantly, we know that $C$ is not influenced, thus deducing the absence of edge $BC$ in the live graph. So we are more certain about which live graph 
is likely to occur than we were in the beginning (having eliminated live graphs containing edge $BC$).
Based on this observation, the algorithm would select 
 $C$ as seed node for second phase (in time step 3), which in turn, would influence its set of 100 nodes in the following time step. Thus the process stops at time step \textbf{4} with \textbf{303} influenced nodes. 
%
%
Note that during its first phase, two-phase diffusion is expected to be slower than the single phase one, because of using only partial budget.

\begin{figure}[t]
\centering
\begin{tabular}{cc}
\includegraphics[scale=.4]{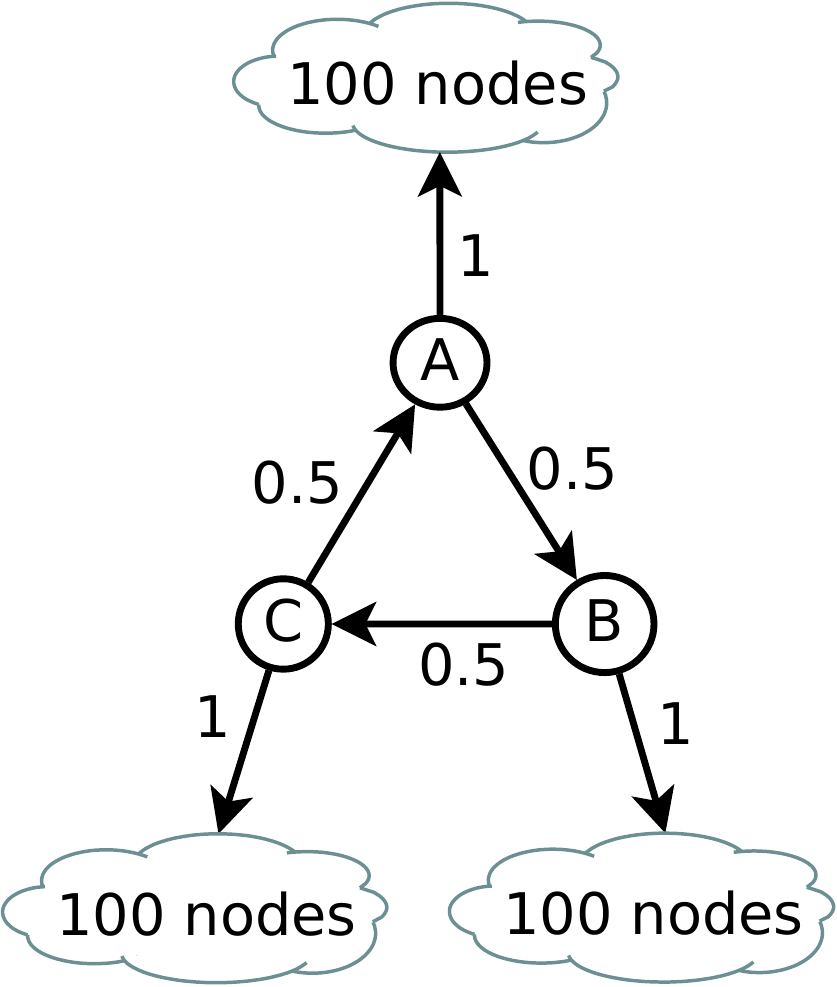}
&
\includegraphics[scale=.4]{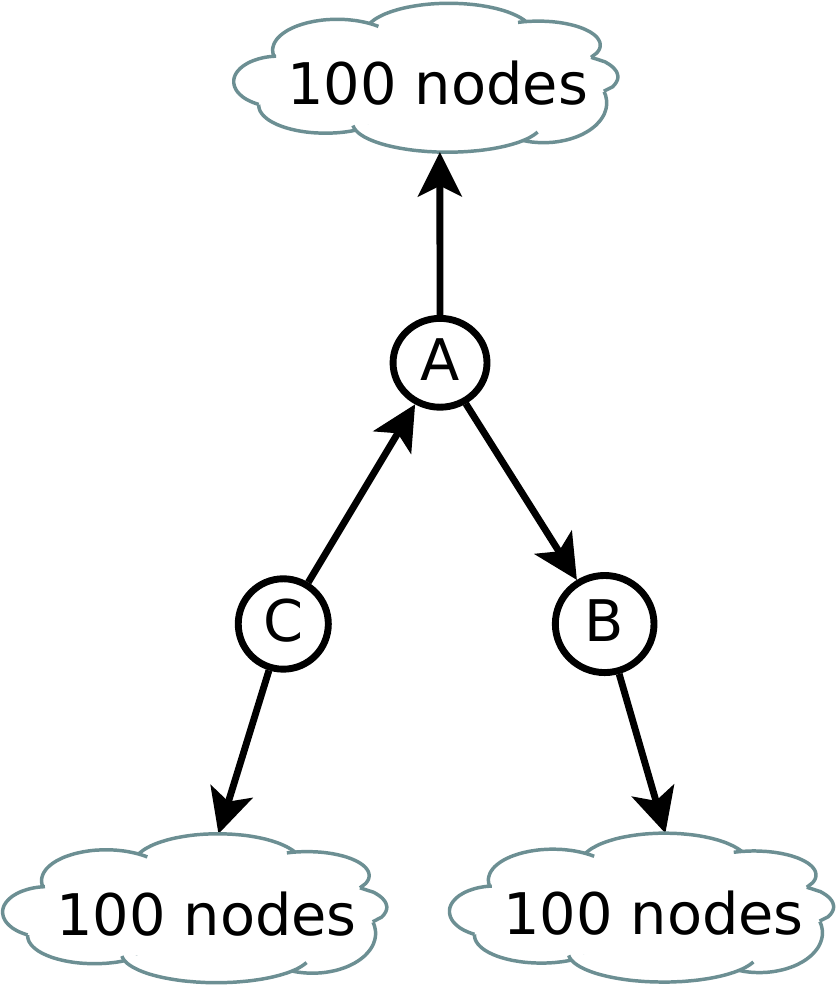}
\\ 
(a) input graph & (b) live graph
\end{tabular}
\caption{\mbox{Multi-phase diffusion: a motivating example}}
\label{fig:motiv_mpid}
\vspace{-2mm}
\end{figure}

If the algorithm had selected $B$ as seed node for~first phase, the diffusion observed after 2 time steps would have guided the algorithm to select $C$ as seed node for second phase, since it would influence $A$ with probability 0.5 (also, given that $B$ is already influenced, selection of $A$ would not influence $C$), thus leading to all 303 nodes getting influenced.
In another case, if $C$ gets selected as seed node for first phase, it would influence all the nodes without having to utilize the entire budget of $2$. So multi-phase diffusion can 
help determine redundancy in seed selection owing to an intermediate check regarding the extent of influence spread. 
Thus it can also
 help achieve a desired spread with a reduced budget.


%

In short, the idea behind two-phase diffusion is~that, for influence maximization algorithms (especially those predicting expected spread over live graphs), reducing the space of possible live graphs results in a better estimate of expected spread, leading to selection of a better seed set.
In fact, two-phase diffusion would facilitate an improvement while using a general influence maximization algorithm, owing to knowledge of already and recently activated nodes, and hence a refined search space for the second phase seed nodes. 

\subsection{The Agenda}
\label{sec:contrib_mpid}

This paper makes the
following specific contributions.

{\scriptsize{$\bullet$}} Focusing on two-phase diffusion in social networks under IC model, we formulate an appropriate objective function 
and investigate its properties. 
We then 
propose an alternative objective function for ease and efficiency of practical implementation. (Section~\ref{sec:problem_mpid})

{\scriptsize{$\bullet$}} We investigate different candidate  algorithms for two-phase influence maximization including extensions of
existing ones that are popular for single phase diffusion. In particular, we propose the use of cross entropy method and a Shapley value based method as promising algorithms for the considered problem.
Seed selection for the two phases could be done in two natural ways: (a) myopic or (b) farsighted. (Section~\ref{sec:algo})

{\scriptsize{$\bullet$}} With extensive simulations on real-world datasets, we study the performance of the proposed algorithms to get an idea how two-phase diffusion would perform, even when used most na\"ively. (Section~\ref{sec:simulations})

{\scriptsize{$\bullet$}} To achieve the best performance out of  two-phase diffusion, we focus on two constituent problems, namely, (a) {\em budget splitting}: how to split the total available budget between the two phases and (b) {\em scheduling}: when to commence the second phase. 
Through a deep investigation of the nature of observations, we propose efficient algorithms for the combined optimization problem of budget splitting, scheduling, and seed sets selection.
We then present key insights from a detailed simulation study. (Section~\ref{sec:practical})

{\scriptsize{$\bullet$}} We conclude the paper with some notes and possible future directions to this work.
(Section~\ref{sec:conclusion_mpid})

\section{Two Phase Diffusion: A Model and Analysis}
\label{sec:problem_mpid}

	Let $k \leq n$ be the total budget, that is, the sum of the number of seed nodes that can be selected in the two phases put together. At time step 0, suppose $k_1$ seed nodes are selected for the first phase and at time step, say $d$, $k_2$ (where $k_2 = k - k_1$) seed nodes are selected for the second phase. 
Our objective is to maximize the expected number of influenced nodes at the end of two-phase diffusion. 
In what follows, we assume $k_1,k_2,d$ to be given.
We study the problem of optimizing over these parameters in Section~\ref{sec:practical}.

\subsection{Objective Function}
\label{sec:objectivefn}

Let $X$ be a live graph obtained by sampling edges for a given graph $G$.
Let $\sigma ^X (S)$ be the number of nodes reachable from seed set $S$ in $X$, that is, the number of nodes influenced at the end of the diffusion process that starts at $S$, if the resulting live graph is $X$. Let $p(X)$ be the probability of occurrence of $X$. So the number of influenced nodes at the end of the process, in expectation, is
$\sigma (S) = \sum_X p(X) \sigma ^X (S)$.
%
It has been shown that $\sigma ^X (S)$, and hence $\sigma (S)$, are non-negative, monotone increasing, and submodular \cite{kempe2003maximizing}.

%
We now formulate an appropriate objective function that measures the expected number of influenced nodes at the end of two-phase diffusion.
Let $S_1$ be the seed set for the first phase
and $X$ be the live graph that is destined to occur ($X$ is not known at the beginning of diffusion, but we know $p(X)$ from edge probabilities in $G$).
Let $Y$ be the partial observation at time step $d$, owing to the observed diffusion. 
As we will be able to classify activated nodes as already activated and recently activated at time step $d$, we assume that $Y$ conveys this information.
That is, from $Y$, the set of already activated nodes $\mathcal{A}^Y$ and the set of recently activated nodes $\mathcal{R}^Y$ at time step $d$, can be determined.
Given $Y$, we can now update the probability of occurrence of a live graph $X$ by $p(X|Y)$.

Now at time step $d$, we should select a seed set that maximizes the final influence spread, considering that nodes in $\mathcal{R}^Y$
will also be effectively acting like seed nodes for second phase.
Let $S_2 ^{O(Y,k_2)}$ be an optimal set of $k_2$ nodes to be selected as seed set, given the occurrence of partial observation $Y$ (which implicitly gives $\mathcal{A}^Y,\mathcal{R}^Y$). 
So
for all $S_2 ' \subseteq N \setminus S_1$ such that $|S_2 ' | \leq k_2$ (it is optimal to have $|S_2 ' | = k_2$ owing to monotone increasing property of $\sigma(\cdot)$),
we have

\vspace{-3mm}
\begin{small}
\begin{equation}
\begin{split}
 \sum_{X} p(X|Y) \sigma^{X \setminus \mathcal{A}^Y} (\mathcal{R}^Y \cup S_2 ^{O(Y,k_2)}) 
 \\
 \geq \sum_{X} p(X|Y) \sigma^{X \setminus \mathcal{A}^Y} (\mathcal{R}^Y \cup S_2 ')
\end{split}
\label{eqn:ineq}
\end{equation}
\vspace{-2mm}
\end{small}

%
\noindent
where ${X \setminus \mathcal{A}^Y}$ is the graph derived from $X$ by removing nodes belonging to $\mathcal{A}^Y$.
%
        Note that we can write $S_2 ^{O(Y,k_2)}$ as $S_2 ^{O(X,S_1,d,k_2)}$, since $Y$ can be uniquely obtained for a given $d$ and particular $X$ and $S_1$. 
        %

%
%
%

So assuming that, given a $Y$, we will select an optimal seed set for the second phase, our objective is to select an optimal $S_1$ (seed set for first phase). Now, as $Y$ is unknown at the beginning of the first phase, the objective function, say $\mathbb{F}(S_1,d,k_2)$,
is an expected value with respect to all such $Y$'s.
Until Section~\ref{sec:practical}, we assume $k_2$ and $d$ to be given, and so we write $\mathbb{F}(S_1,d,k_2)$ as $f(S_1)$.
So,
$f(S_1)$ is

~
\vspace{-5mm}
\begin{small}
\begin{align}
&\;\;
\sum_{Y} p(Y)  \big\{ |\mathcal{A}^{Y}| + \sum_{X} p(X|Y) \sigma^{X \setminus \mathcal{A}^Y}(\mathcal{R}^Y \cup S_2 ^{O(Y,k_2)})   \big\}
\label{eqn:first}
\\ &\hspace{-2mm}= 
\sum_{Y} p(Y)  \big\{ |\mathcal{A}^{Y}| + \sum_{X} p(X|Y) \sigma^{X \setminus \mathcal{A}^Y}(\mathcal{R}^Y \cup S_2 ^{O(X,S_1,d,k_2)})   \big\}
\nonumber
\\ &\hspace{-2mm}= 
\sum_{Y} p(Y) \sum_{X} p(X|Y) \big\{ |\mathcal{A}^{Y}| +  \sigma^{X \setminus \mathcal{A}^Y}(\mathcal{R}^Y \cup S_2 ^{O(X,S_1,d,k_2)})   \big\}
\nonumber
\\
&\;\;\;\;\;\;\;\;\;\;\;\;\;\;\;\;\;\;\big(\because |\mathcal{A}^{Y}| = |\mathcal{A}^{Y}|\sum_{X} p(X|Y) = \sum_{X} p(X|Y)|\mathcal{A}^{Y}| \big)
\nonumber
\\ &\hspace{-2mm}= 
\sum_{Y} p(Y) \sum_{X} p(X|Y) \sigma^X (S_1 \cup S_2 ^{O(X,S_1,d,k_2)}) 
\label{eqn:basic}
\\ &\hspace{-2mm}= 
\sum_{X} \sum_{Y} p(Y) p(X|Y) \sigma^X (S_1 \cup S_2 ^{O(X,S_1,d,k_2)})  
\nonumber
\end{align}
%
\begin{equation}
\label{eqn:f}
\therefore \; f(S_1) = \sum_X p(X) \sigma^X (S_1 \cup S_2 ^{O(X,S_1,d,k_2)}) 
\end{equation}
\vspace{-2mm}
\end{small}

%
Note that at time step $d$, the choice of $S_2 ^{O(X,S_1,d,k_2)}$ depends 
not only on $X$, but
on partial observation $Y$, and hence on all live graphs that could result from $Y$ (just as in single phase, choice of the best seed set depends on all live graphs that could result from the given graph $G$).
It is easy to prove on similar lines as \cite{kempe2003maximizing} that, the problem of maximizing $f(\cdot)$ is NP-hard.

On both sides of Inequality (\ref{eqn:ineq}), adding $|\mathcal{A}^{Y}|$ and taking convex combination over all $Y$'s, LHS transforms to Expression (\ref{eqn:first}) which we have shown to be equivalent to Expression (\ref{eqn:f}). Transforming RHS in Inequality (\ref{eqn:ineq}) on similar lines, we have

\vspace{-2mm}
\begin{small}
\begin{displaymath}
\sum_{X} p(X) \sigma^X (S_1 \cup S_2 ^{O(X,S_1,d,k_2)})  \geq \sum_{X} p(X) \sigma^X (S_1 \cup S_2 ')
\end{displaymath}
\vspace{-2mm}
\end{small}

\noindent
We call this inequality, the {\em optimality of} $S_2 ^{O(X,S_1,d,k_2)}$. 
%

\noindent
We now show how to compute $f(\cdot)$ using an example.

\vspace{-5mm}
\begin{wrapfigure}{r}{0.16\textwidth}
\centering
\vspace{2mm}
\includegraphics[scale=.33]{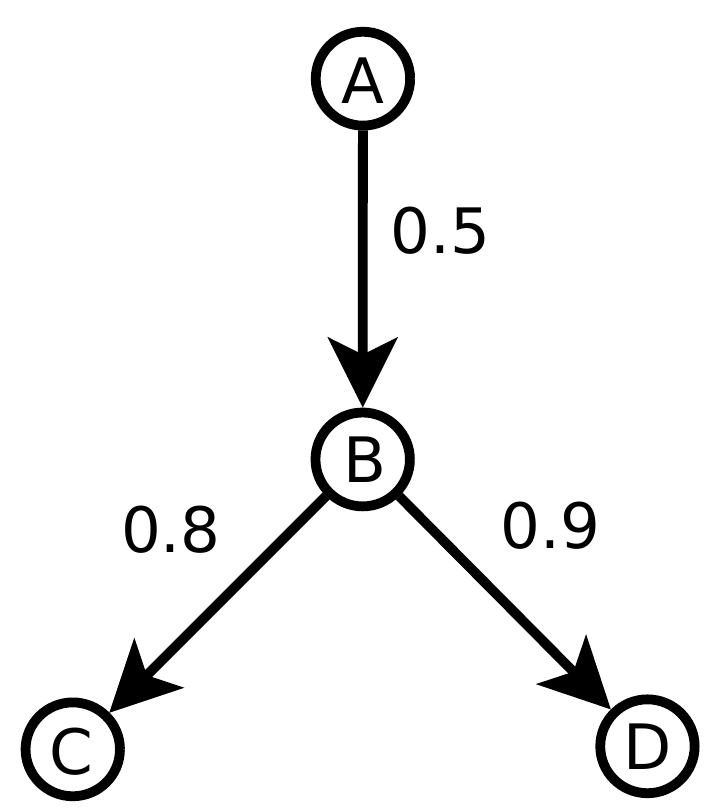}
\caption{Example
}
\label{fig:submod_counter}
\vspace{-.3cm}
\end{wrapfigure}
\noindent
\begin{example}
\label{eg:obj_fn}
%
(Figure~\ref{fig:submod_counter})
Consider 
$S_1=\{A\}$, $k_2=1$, and $d=1$.
The table below
lists the two possibilities of $Y$ ($\mathcal{A}^Y,\mathcal{R}^Y$) at $d=1$. 
The set $S_2^{O(Y,k_2)}$ is easy to compute for both the cases.
So $f(\{A\})= \mathbb{F}(\{A\},1,1) = \sum_X p(X) \sigma^X (\{A\} \cup S_2 ^{O(X,\{A\},1,1)}) = 3.8$.
\end{example}

\begin{center}
\vspace{3mm}
\begin{footnotesize}
\begin{tabular}{c|c|c|c|c|c}
\hline \hline  
\multicolumn{5}{c}{\T \B ~~~~~~~~~~~~~~~~~~$S_1=\{A\}, k_2=1, d=1$} \\
\hline \T \B \hspace{-.4cm}
  \multirow{3}{*}{$X$} & \multirow{3}{*}{$p(X)$} & \multicolumn{2}{c|}{\T \B $Y$} & \multirow{3}{*}{$S_2^{O(Y,k_2)}$}  & \multirow{3}{*}{$f(S_1)$}  \\ 
  \cline{3-4} \hspace{-.4cm} 
 & & \multirow{2}{*}{$\mathcal{A}^Y$} & \multirow{2}{*}{$\mathcal{R}^Y$}  & & \\
 \hspace{-.4cm} 
 & & & & & \\
  \hline \hline \T \B  \hspace{-.4cm}
  $\{AB,BC,BD\}$ & 0.36 & \multirow{4}{*}{$\{A\}$}  & \multirow{4}{*}{$\{B\}$} &  \multirow{4}{*}{$\{C\}$} & 4\\
   \cline{1-2}\cline{6-6} \hspace{-.4cm} \T \B
  $\{AB,BC\}$  & 0.04 &  & &   & 3 \\
   \cline{1-2}\cline{6-6} \hspace{-.4cm} \T \B
     $\{AB,BD\}$ & 0.09 & & &  & 4 \\
   \cline{1-2}\cline{6-6} \hspace{-.4cm} \T \B
     $\{AB\}$ & 0.01 &  & &  & 3\\
  \hline \T \B \hspace{-.4cm}
  $\{BC,BD\}$ &  0.36 & \multirow{4}{*}{$\{A\}$} & \multirow{4}{*}{$\{\}$} & \multirow{4}{*}{$\{B\}$} 
  &  4 \\
  \cline{1-2}\cline{6-6} \hspace{-.4cm} \T \B
   $\{BC\}$ & 0.04  &  &  & & 3\\
   \cline{1-2}\cline{6-6} \hspace{-.4cm} \T \B
   $\{BD\}$ & 0.09  & &  & & 3\\
   \cline{1-2}\cline{6-6} \hspace{-.4cm} \T \B
   $\{\}$ & 0.01  & &  & & 2\\ 
   \hline \hline
\end{tabular}
\end{footnotesize}
\end{center}


\subsection{Properties of the Objective Function}
\label{sec:props}

\begin{property}
\label{prop:monotone}
$f(\cdot)$ is non-negative and monotone increasing.
\end{property}
\begin{proof}
$f(\cdot)$ is non-negative since $\sigma ^X (\cdot)$ is non-negative.
Consider $S_1 \subset T_1$. Then,

~
\vspace{-5mm}
\begin{small}
\begin{align*}
f(T_1)
&=
\; \sum_{X} p(X) \sigma^X (T_1 \cup T_2 ^{O(X,T_1,d,k_2)})
\\ & \geq
\; \sum_{X} p(X)   \sigma^X (T_1 \cup S_2 ^{O(X,S_1,d,k_2)})
   \\ & \geq 
\; \sum_{X} p(X)   \sigma^X (S_1 \cup S_2 ^{O(X,S_1,d,k_2)})
\;   =f(S_1)
\end{align*}
\vspace{-2mm}
\end{small}

\noindent
The first inequality is from optimality of $T_2 ^{O(X,T_1,d,k_2)}$  and the second one from monotonicity of $\sigma ^X (\cdot)$.
\end{proof}

\noindent
As $f(\cdot)$ is monotone increasing and $|S_1| \leq k_1$, given a fixed $k_1$, it is optimal to select $S_1$ such that $|S_1|=k_1$.

\begin{property}
\label{prop:nonsubmodular}
{$f(\cdot)$ is neither submodular nor supermodular.}
\end{property}
\begin{proof}
We prove this using a simple counterexample network in Figure~\ref{fig:submod_counter}.
Consider $d = 3$ and $k_2 = 1$. 

Considering $S_1 = \{\}$, $T_1 = \{D\}$, $i = C$, 
we get $f(S_1 \cup \{i\} ) = 2.95$, $f(S_1) = 2.7$, $f(T_1 \cup \{i\} ) = 3.5$, $f(T_1) = 2.9$.
So we have 
$f(S_1 \cup \{i\} ) - f(S_1) < f(T_1 \cup \{i\} ) - f(T_1)$
for some $T_1$, $S_1 \subset T_1$,  $i \notin T_1$, which proves that  $f(\cdot)$ is not submodular.

Considering $S_1 = \{\}$, $T_1 = \{B\}$, $i = A$, 
we get $f(S_1 \cup \{i\} ) = 3.84$, $f(S_1) = 2.7$, $f(T_1 \cup \{i\} ) = 3.98$, $f(T_1) = 3.7$.
So we have 
$f(S_1 \cup \{i\} ) - f(S_1) > f(T_1 \cup \{i\} ) - f(T_1)$
for some $T_1$, $S_1 \subset T_1$,  $i \notin T_1$, which proves that  $f(\cdot)$ is not supermodular.
\end{proof}

\begin{remark}
Simulations on test graphs showed the satisfiability of the diminishing returns property in most cases, that is, for most $S_1,T_1,i$ such that $S_1 \subset T_1 \subset N$ and  $i \in N\setminus T_1$, $f(S_1 \cup \{i\} ) - f(S_1) \geq f(T_1 \cup \{i\} ) - f(T_1)$.
\end{remark}

\begin{property}
\label{prop:subadditive}
\mbox{$f(\cdot)$ is subadditive.}
\end{property}
\begin{proof}
Let $V_1=S_1\cup T_1$
 and $V_2^{O(X,V_1,d,k_2)}$ be an optimal set of $k_2$ nodes given the observation corresponding to $X,d$ starting with seed set $V_1$.
 So,

 \vspace{-2mm}
 \begin{small}
\begin{align*}
& \hspace{-1mm}  f(S_1)+f(T_1) \\
& \hspace{-2mm} =
\sum_{X} p(X) \{ \sigma^X (S_1 \cup S_2 ^{O(X,S_1,d,k_2)}) 
+ \sigma^X (T_1 \cup T_2 ^{O(X,T_1,d,k_2)}) \}
\\& \hspace{-2mm} \geq
\sum_{X} p(X) \{ \sigma^X (S_1 \cup V_2 ^{O(X,V_1,d,k_2)}) 
+ \sigma^X (T_1 \cup V_2 ^{O(X,V_1,d,k_2)}) \}
\\& \hspace{-2mm} \geq
\sum_{X} p(X) \sigma^X (S_1 \cup T_1 \cup V_2 ^{O(X,V_1,d,k_2)})
\\& \hspace{-2mm} =
\sum_{X} p(X) \sigma^X (V_1 \cup V_2 ^{O(X,V_1,d,k_2)})
=
f(V_1) = f(S_1 \cup T_1)
\end{align*}
\vspace{-2mm}
\end{small}

\noindent
The first inequality is from optimality of $S_2 ^{O(X,S_1,d,k_2)}$ and $T_2 ^{O(X,T_1,d,k_2)}$, and the second one from subadditivity of $\sigma ^X (\cdot)$ (since submodularity and non-negativity  implies subadditivity).
\end{proof}

Owing to the sequential nature of the considered problem, dynamic programming seems to be a natural approach. However, there are two major issues with its usage. Owing to NP-hardness of the single phase influence maximization problem, it is impractical to compute $S_2^{O(X,S_1,d,k_2)}$ in Equation~(\ref{eqn:f}). So the first issue is that finding an optimal solution to a subproblem itself is computationally infeasible. Second, the number of possible subproblems is exponential in the number of nodes (different $S_1$'s would almost certainly result in different $Y$'s, for any given $X$ and $d$). So even if one solves the subproblem approximately, 
the probability of reusing the stored solutions is negligible.
Also
as stated earlier, there exists an approximation algorithm for maximizing a subadditive function
\cite{feige2009maximizing}. However, owing to its relatively high running time, we leave it out of our study. It would be interesting though to develop more efficient algorithms that could exploit the subadditivity of $f(\cdot)$.

As mentioned earlier,
it is impractical to compute $S_2^{O(X,S_1,d,k_2)}$ in Equation~(\ref{eqn:f}). 
We surmount this difficulty by maximizing an alternative function instead of $f(\cdot)$. To emphasize this point, note that this impractical computation is for finding the objective function value itself, which makes finding an optimal $S_1$, a computationally infeasible task. So the alternative function must be several orders of magnitude faster to compute than $f(\cdot)$.
We now address this problem.

\subsection{An Alternative Objective Function}
\label{sec:compute}

\subsubsection{Using Greedy Hill-climbing Algorithm}
\label{sec:greedy}

Given the occurrence of the partial observation $Y$,
let $S_2 ^{G(Y,k_2)} = S_2 ^{G(X,S_1,d,k_2)}$ be a set of size $k_2$ obtained using the greedy hill-climbing algorithm.
%
Let

\vspace{-1mm}
\begin{small}
\begin{displaymath}
\label{eqn:g}
g(S_1) \overset{\mathcal{MC}}=
\sum_X p(X) \sigma^X (S_1 \cup S_2 ^{G(X,S_1,d,k_2)}) 
\end{displaymath}
\vspace{-3mm}
\end{small}

\noindent
(${\mathcal{MC}}$ means obtained using Monte-Carlo simulations).

\begin{theorem} 
\label{thm:nemhauser}
For a non-negative, monotone increasing, submodular function $\mathcal{F}$, let $S^G$ be a set of size $k$ obtained using greedy hill-climbing. Let $S^O$ be a set that maximizes the value of $\mathcal{F}$ over all sets of size $k$. Then 
for any $\epsilon>0$, there is a $\gamma>0$ such that by using $(1 + \gamma)$-approximate values for $\mathcal{F}$, we obtain a $(1-\frac{1}{e}-\epsilon)$-approximation~\cite{kempe2003maximizing}.
\end{theorem}

So the greedy hill-climbing algorithm would provide an approximation guarantee that is arbitrarily close to $(1- \frac{1}{e})$ for maximizing such a function,
if we are able to compute the function value with sufficient accuracy.
The $(1 + \gamma)$-approximate values for $\mathcal{F}$ with small $\gamma$ can be obtained using sufficiently large number of Monte-Carlo simulations 
\cite{kempe2003maximizing}.

\begin{lemma}
\label{lem:f_approx_g}
$g(\cdot)$ gives a $\left( 1-\frac{1}{e}-\epsilon \right)$ approximation to $f(\cdot)$.
\end{lemma}
\begin{proof}
Let $\Phi_T (S) = \sigma (T \cup S)$ and $\Phi_T^X (S) = \sigma ^X (T \cup S)$.
It can be easily shown that $\Phi_T ^X (S)$, and hence $\Phi_T (S)$, are non-negative, monotone increasing, and submodular.
%
So we have

\vspace{-3mm}
\begin{small}
\begin{align*}
g(S_1)
&\overset{\mathcal{MC}}{=} 
\sum_X p(X) \Phi_{S_1}^X (S_2 ^{G(X,S_1,d,k_2)}) 
 \\ & \geq 
\text{\scriptsize{$\left( 1-\frac{1}{e}-\epsilon \right)$}}  \sum_X p(X) \Phi_{S_1}^X (S_2 ^{O(X,S_1,d,k_2)}) 
 \\ &=
\text{\scriptsize{$ \left( 1-\frac{1}{e}-\epsilon \right)$}} f(S_1) 
\end{align*}
\vspace{-3mm}
\end{small}

\noindent
where the first inequality results from Theorem~\ref{thm:nemhauser}.
\end{proof}

So one can aim to maximize $g(\cdot)$ instead of $f(\cdot)$.
However, greedy hill-climbing algorithm itself is expensive in terms of running time (even after optimizations such as in \cite{chen2009efficient}), so we aim to maximize yet another function which would act as a proxy for $g(\cdot)$.

\subsubsection{Using Generalized Degree Discount Heuristic}
\label{sec:gdd}
Consider the process of selecting seed nodes one at a time.
At a given time in the midst of the process, let $\mathcal{X}=$ set of in-neighbors of node $v$ already selected as seed nodes and $\mathcal{Y}=$ set of out-neighbors of $v$ not yet selected as seed nodes.
We develop Generalized Degree Discount (GDD) Heuristic
as an extension to 
the argument for Theorem~2 in \cite{chen2009efficient}: 
if $v$ is not (directly) influenced by any of the already selected seeds, which occurs with probability
{\small $ \prod_{x \in \mathcal{X}} (1-p_{xv}) $}, then the additional expected number of nodes that it influences directly (including itself) is {\small $\big(1+\sum_{y \in \mathcal{Y}} p_{vy}\big)$}.
So until the budget is exhausted, GDD heuristic iteratively selects a node $v$ having the largest value of

\vspace{-1mm}
\begin{small}
\begin{equation}
\label{eqn:gdd_value}
w_v = \Big( \prod_{x \in \mathcal{X}} (1-p_{xv}) \Big) \Big( 1+\sum_{y \in \mathcal{Y}} p_{vy} \Big)
\end{equation}
\vspace{-3mm}
\end{small}

Its time complexity is $ O( k n \Delta )$, where $\Delta$ is the maximum degree in the graph.

Given the occurrence of the partial observation $Y$,
let $S_2 ^{H(Y,k_2)} = S_2 ^{H(X,S_1,d,k_2)}$ be a set of size $k_2$ obtained using the 
GDD heuristic.
Let 

\vspace{-1mm}
\begin{small}
\begin{equation}
\label{eqn:h}
h(S_1) \overset{\mathcal{MC}}{=} 
\sum_X p(X) \sigma^X (S_1 \cup S_2 ^{H(X,S_1,d,k_2)}) 
\end{equation}
\vspace{-3mm}
\end{small}

%
\noindent
We conducted simulations for checking how well $h(\cdot)$ acts as a proxy for $g(\cdot)$, using both weighted cascade and trivalency models. We observed the following. 
\begin{observation}
\label{obs:gdd_approx_greedy}
For almost all $S,T$ pairs: \\
{\scriptsize{$\bullet$}} If $g(T)>g(S)$, then $h(T)>h(S)$ (in particular, this is satisfied for almost all pairs of sets that give excellent objective function values), which ensures that the selected seed set remains unchanged in most cases when we have $h(\cdot)$ as our objective function instead of $g(\cdot)$.\\
{\scriptsize{$\bullet$}} $\frac{h(S)}{h(T)} \approx \frac{g(S)}{g(T)}$, which ensures that the seed set selected by algorithms, which implicitly depend on the ratios of the objective function values given by any two sets, remains unchanged in most cases when we have $h(\cdot)$ as our objective function instead of $g(\cdot)$; two of the algorithms we consider, namely, FACE (Section~\ref{sec:ce_method}) and SPIC (Section~\ref{sec:shapley_method}) belong to this category of algorithms. 
\end{observation}

\begin{remark}
One could question, why not use a function $\hat h(S_1) \overset{\mathcal{MC}}{=}  \sum_X p(X) \sigma^X (S_1 \cup S_2 ^{\hat H(X,S_1,d,k_2)})$ instead of $h(S_1)$, where $S_2 ^{\hat H(X,S_1,d,k_2)}$ is a set of size $k_2$ obtained using PMIA algorithm (it has been observed to perform very close to greedy algorithm on practically all relevant datasets, while running orders of magnitude faster \cite{chen2010scalable}). However, though PMIA is an efficient algorithm for single phase influence maximization, it is highly undesirable to use it for computation of objective function value alone (since an algorithm designed to maximize $\hat h(\cdot)$ would require computation of function values for a large number of sets). On the other hand, GDD is orders of magnitude faster than PMIA. Though we use moderately sized datasets for making Observation~\ref{obs:gdd_approx_greedy}, we could stretch the size of datasets by aiming to observe how well $h(\cdot)$ acts as a proxy for $\hat h(\cdot)$.
\end{remark}

Owing to the above justifications, we aim to maximize $h(\cdot)$ instead of $f(\cdot)$, for two-phase influence maximization in the rest of this paper.

\section{Algorithms for Two-Phase Influence Maximization}
\label{sec:algo}

In the previous section, we formulated the objective function for two-phase influence maximization $f(\cdot)$ and studied its properties. In addition to their theoretical relevance, these properties have implications for as to which algorithms are likely to perform well. We present them while describing the algorithms.

%
Let $\mathcal{T}$ be the time taken for computing the objective function value for a set.\\
%
{\scriptsize{$\bullet$}} For $\sigma(\cdot)$, $\mathcal{T} = O(m \mathcal{M})$,
where $\mathcal{M}$ is the number of Monte-Carlo simulations.\\
{\scriptsize{$\bullet$}} For $h(\cdot)$, $\mathcal{T} = O(k_2 n \Delta m \mathcal{M}_1 \mathcal{M}_2)$,
where $\mathcal{M}_1$ and $\mathcal{M}_2$ are the numbers of Monte-Carlo simulations for first and second phases, respectively, and $\Delta$ is the maximum out-degree in the graph.

\subsection{Candidate Algorithms for Seed Selection}

Now we present the algorithms that we consider for seed selection for single phase influence maximization, 
which we later extend to the two-phase case.

\vspace{-1mm}
\subsubsection{Greedy Algorithm}

As described earlier, the greedy (hill-climbing) algorithm for maximizing a function $\mathcal{F}$, selects nodes one at a time, each time choosing a node that provides the largest marginal increase in the value of $\mathcal{F}$, until the budget is exhausted.
Its time complexity is $O(kn\mathcal{T})$.
As noted earlier, though our two-phase objective function is not submodular, we observed that the diminishing returns property was satisfied in most cases; so even though the condition in Theorem~\ref{thm:nemhauser} is not satisfied, the greedy algorithm is likely to perform well.
Further, unlike in the case of single phase influence maximization, we cannot use CELF optimization for the two-phase case owing to its objective function being non-submodular
(though satisfiability of the diminishing returns property in most cases may make it a reasonable approach,
we do not use it so as
to preserve performance accuracy of the greedy algorithm for two-phase influence maximization).


\vspace{-1mm}
\subsubsection{Single/Weighted Discount Heuristics (SD/WD)}

The single discount (SD) heuristic~\cite{chen2009efficient} for a graph $G$ can be described as follows: select the node having the largest number of outgoing edges in $G$, then remove that node and all of its incoming and outgoing edges to obtain a new graph $G'$, again select the node having the largest number of outgoing edges in the new graph $G'$, and continue until the budget is exhausted.
%
Weighted discount (WD) heuristic is a variant of SD heuristic where, sum of outgoing edge probabilities is considered instead of number of outgoing edges.
The time complexity of these heuristics is $ O( k n \Delta )$.
These heuristics run extremely fast and hence can be used for efficient seed selection for very large networks. 


\vspace{-1mm}
\subsubsection{Generalized Degree Discount (GDD) Heuristic}
The generalized degree discount (GDD) heuristic is as described in Section~\ref{sec:gdd}.


\vspace{-1mm}
\subsubsection{PMIA}
\label{sec:pmia}
This heuristic, based on the arborescence structure of influence spread,
is shown to perform close to greedy algorithm and runs orders of magnitude faster \cite{chen2010scalable}.

\vspace{-1mm}
\subsubsection{Fully Adaptive Cross Entropy Method (FACE)}
\label{sec:ce_method}

It has been shown that the cross entropy (CE) method provides an efficient and general method for solving combinatorial optimization problems~\cite{de2005tutorial}.
In our context, the CE method involves an iterative procedure where each iteration consists of two steps, namely,
(a) generating data samples (a vector consisting of a sampled 
candidate seed set) according to a specified distribution and
(b) updating the distribution based on the sampled data to produce better samples in the next iteration.
 We use an adaptive version of the CE method called the {\em fully adaptive cross entropy} (FACE) algorithm~\cite{de2005tutorial}.
 Its time complexity is $O(n\mathcal{T} \mathcal{I})$, where $\mathcal{I}$ is the number of iterations taken for the algorithm to terminate.
 However, the running time can be drastically reduced for single phase diffusion using preprocessing similar to that for greedy algorithm as in \cite{chen2009efficient}.
An added advantage of this algorithm is that it would not only find an optimal seed set, but also implicitly determine how to split the total budget between the two phases and also the delay after which the second phase should be triggered (see Section~\ref{sec:practical}).
%

\vspace{-1mm}
\subsubsection{Shapley Value based - IC Model (SPIC)}
\label{sec:shapley_method}

We consider a Shapley value based method because it is shown to perform well even when the  objective function is non-submodular~\cite{narayanam2010shapley}.
It has been observed that,
in order to obtain the seed nodes after computing Shapley values of the nodes, 
some post-processing is required.
%
%
As the post-processing step under the IC model, we propose the following discounting scheme: 
\\
(a) Since node $x$ would get directly activated because of node $y$ with probability $p_{yx}$, we discount the value of $x$ by multiplying it with $(1-p_{yx})$ whenever any of its in-neighbors $y$ gets chosen in the seed set. 
%
\\
(b) As node $z$ influences node $y$ directly with probability $p_{zy}$, it gets a fractional share of $y$'s value (since $z$ would be influencing other nodes indirectly, through $y$).
So when $y$ is chosen in the seed set, we subtract $y$'s share ($p_{zy}\phi_y$ where $\phi_y$ is the value of $y$ during its selection) from the current value of $z$.
If the value becomes negative because of oversubtraction, we assign zero value to it.
%
%
\\
A node, not already in the seed set, with the highest value after discounting, is then added to the seed set in a given iteration.
In our simulations, we observed that this discounting scheme outperforms the SPIN algorithm (choosing seed nodes one at a time while eliminating neighbors of already chosen nodes~\cite{narayanam2010shapley}).
Assuming $O(n)$ permutations for approximate computation of Shapley value~\cite{narayanam2010shapley} (since exact computation is \#P-hard),
the algorithm's time complexity is approximately 
$O(n \mathcal{T})$.
It is to be noted, however, that the SPIN algorithm~\cite{narayanam2010shapley} is not scalable to very large networks even for single phase influence maximization~\cite{chen2010scalable}, and so isn't SPIC.

\vspace{-1mm}
\subsubsection{Random Sampling and Maximizing (RMax)}

Here, 
we sample $O(n)$ number of sets that satisfy the budget constraint, 
and then assign that set as the seed set which gives the maximum function value among the sampled sets.
Note that this method is different from the random set selection method~\cite{kempe2003maximizing}, 
where only one sample is drawn.
Its time complexity is $O(n\mathcal{T})$.
We consider this method as it is very generic and agnostic to the properties of the objective function, and can be used for optimizing functions with arbitrary or no structure. This method is likely to perform well when the number of samples is sufficiently large.

\subsection{Extension of Algorithms to Two-phase Influence Maximization}

Now we present how the aforementioned single phase influence maximization algorithms can be extended for two-phase influence maximization.
Let $\mathcal{F}_1 (\cdot)$ and $\mathcal{F}_2 (\cdot)$ be objective functions corresponding to 
seed selection in
 first and second phases, respectively.
Consider an influence maximization algorithm $\mathbb{A}$.

We explore two special cases of 
Algorithm 1 
(the notation can be recalled from Sections~\ref{sec:objectivefn} and \ref{sec:compute}):
\begin{tabbing}
1.  Farsighted \= : $\mathcal{F}_1 (S_1) = h(S_1) \, , \,\mathcal{F}_2 (S_2) = \sigma(\mathcal{R}^Y \cup S_2)$ \\
2.   Myopic \> : $\mathcal{F}_1 (S_1) = \sigma(S_1) \, , \,\mathcal{F}_2 (S_2) = \sigma(\mathcal{R}^Y \cup S_2)$ 
\end{tabbing}

\newpage
%
\hrule
\vspace{.1cm}
\noindent\textbf{Algorithm 1} Two-phase general algorithm (IC model)
\hrule
\vspace{.1cm}
\begin{small}
\begin{algorithmic}[1]
\renewcommand{\algorithmicrequire}{\textbf{Input:}}
\renewcommand{\algorithmicensure}{\textbf{Output:}}

\REQUIRE $G = (N, E, \mathcal{P})$, $k_1$, $k_2$, $d$

\textbf{First phase:} 

\STATE Find set of size $k_1$ using $\mathbb{A}$ for maximizing $\mathcal{F}_1 (\cdot)$ on $G$

\STATE Run the diffusion using IC model until time $d$ 

\textbf{Second phase:} 

\STATE On observing $Y$, construct $G^d$ from $G$ by deleting $\mathcal{A}^Y$

\STATE With $\mathcal{R}^Y$ forming partial seed set, find set of size $k_2$ using $\mathbb{A}$ for maximizing $\mathcal{F}_2 (\cdot)$ on $G^d$

\STATE Continue running the diffusion using IC model until no further nodes can be influenced

\ENSURE Seed nodes for the first and second phases at time steps 0 and $d$, respectively 

\end{algorithmic}
\end{small}
\vspace{-.3cm}
\hrulefill

As explained earlier, the second phase objective function assumes that $\mathcal{R}^Y$ forms a partial seed set, hence the above form of $\mathcal{F}_2(\cdot)$.
The farsighted objective function looks ahead and accounts for the fact that there is going to be a second phase and hence attempts to maximize $h(\cdot)$, while the myopic function does not.
Note that heuristics such as PMIA, GDD, WD, SD do not consider the actual objective function for seed selection, and so the myopic and farsighted algorithms are the same for these heuristics.

We now formally prove the effectiveness of two-phase diffusion for influence maximization.

\begin{theorem}
For any given values of $k_1$ and $k_2$, the expected influence achieved using optimal two-phase algorithm is at least as much as that achieved using optimal single phase one. 
\end{theorem}
\begin{proof}
Let $S^*$ be the optimal seed set of cardinality $k=k_1+k_2$ selected in single phase diffusion. Let sets $S_1^*$ and $S_2^*$ be such that $|S_1|=k_1$, $|S_2|=k_2$, $S^*=S_1^* \cup S_2^*$, and $S_1^* \cap S_2^* =\emptyset$. 
Now assuming 
any $d$, it is clear from the optimality of $S_2 ^{O(X,S_1^*,d,k_2)}$ (see derivation of $f(\cdot)$ preceding Equation~(\ref{eqn:f})) that,

\vspace{-2mm}
\begin{small}
\begin{align*}
\sum_X p(X) \sigma^X (S_1^* \cup S_2 ^{O(X,S_1^*,d,k_2)}) \geq \sum_X p(X) \sigma^X (S_1^* \cup S_2^* ) 
\end{align*}
\vspace{-2mm}
\end{small}

\noindent
Note that 
%
the left hand side is $f(S_1^*)$ (Equation~(\ref{eqn:f})) and right hand side is $\sigma(S^*)$. So we have,

\vspace{-2mm}
\begin{small}
\begin{align*}
\max_{S_1} f(S_1) \geq f(S_1^*) \geq \sigma(S^*)
\end{align*}
\vspace{-2mm}
\end{small}

\noindent
The leftmost
and rightmost expressions are the expected spreads using 
 two-phase
and single phase optimal algorithms, respectively, hence the result.
Note that this holds for any $d$.
\end{proof}

\section{A Study to Demonstrate Efficacy of Two-phase Diffusion}
\label{sec:simulations}

In this section, we study how much improvement one can expect by diffusing information in two phases over a social network, even with a simple approach. To start with, we assume that $k_1,k_2,d$ are known and our objective is to find the seed sets for the two phases (we study the problem of optimizing over these parameters in Section~\ref{sec:practical}). 
As a simple and na\"ive first approach, we consider an equal budget split between the two phases, that is, $k_1=k_2=\frac{k}{2}$. 
Furthermore, we consider $d=D$, where $D$ is the length of the longest path in the network, so that by time step $D$, the first phase would have completed its diffusion. In practice, $D$ could be the maximum delay that we are ready to incur in absence of any temporal constraints.
Intuitively, it is clear that one should wait for as long as possible before selecting the seed nodes for second phase, as it would give a larger observation and a reduced search space. We now prove this formally.

\begin{lemma}
\label{lem:high_d_better}
For any given values of $k_1$ and $k_2$, the number of nodes influenced using an optimal two-phase influence maximization algorithm is a non-decreasing function of $d$.
\end{lemma}
\begin{proof}
Starting from a given first phase seed set $S_1$, let 
$Y_i$'s be the partial observations at time step $d$.
Also, let $Y_{ij}$'s be the partial observations at time step $d^+ > d$ resulting from a given $Y_i$ at time step $d$. 
%
%
By enumerating the partial observations at time step $d$, the expected number of nodes influenced at the end of diffusion, as given in Equation~(\ref{eqn:basic}), can be written~as

\vspace{-2mm}
\begin{small}
\begin{align*}
&
\sum_{i} p(Y_i) \sum_{X} p(X|Y_i) \sigma^X (S_1 \cup S_2 ^{O(Y_i,k_2)}) 
\\&=
\sum_i \sum_{j} p(Y_{ij}) \sum_{X} p(X|Y_{ij}) \sigma^X (S_1 \cup S_2 ^{O(Y_i,k_2)}) 
\\&\leq
\sum_i \sum_{j} p(Y_{ij}) \sum_{X} p(X|Y_{ij}) \sigma^X (S_1 \cup S_2 ^{O(Y_{ij},k_2)}) 
\end{align*}
\vspace{-2mm}
\end{small}

\noindent
which is the expected number of nodes influenced at the end of diffusion, if the second phase starts at time step $d^+ > d$.
The last inequality results from the optimality of $S_2 ^{O(Y_{ij},k_2)}$ for partial observation $Y_{ij}$.
\end{proof}

The following result now follows directly.

\begin{theorem}
\label{thm:D_is_opt}
For any given values of $k_1$ and $k_2$, the number of nodes influenced using an optimal two-phase influence maximization algorithm is maximized when $d=D$.
\end{theorem}

\begin{remark}
Determining $D$ exactly may be infeasible in practice. For instance, checking whether the first phase has completed its diffusion requires polling at every time step. Also, finding the length of the longest path in the network is known to be an NP-hard problem. 
However, for all practical purposes, $D$ can be approximated by a large enough value based on the network in consideration.
\end{remark}

\subsection{Simulation Setup}

For computing the objective function value and evaluating performance using single phase diffusion, we ran $10^4$ Monte-Carlo iterations (standard in the literature). To set a balance between running time and variance, we ran $10^3$ Monte-Carlo iterations for each of the phases in two-phase diffusion (equivalent to $10^6$ live graphs); the observed variance was negligible. 

As mentioned earlier, for transforming an undirected and unweighted network (dataset) into a directed and weighted network for studying the diffusion process, we consider two popular, well-accepted special cases of the IC model, namely, the {weighted cascade (WC) model} and the {trivalency (TV) model}.
We first conduct simulations on the Les Miserables (LM) dataset  \cite{knuth1993stanford} consisting of 77 nodes and 508 directed edges in order to study the performances of computationally intensive farsighted algorithms for two-phase influence maximization.
%
%
%
%
For studying two-phase diffusion on a larger dataset, we consider an academic collaboration network obtained from co-authorships in the ``High Energy Physics - Theory'' papers published on the e-print arXiv from 1991 to 2003. It contains 15,233 nodes and 62,774 directed edges, and is popularly denoted as NetHEPT. This network exhibits many structural features of large-scale social networks  and is widely used for experimental justifications, for example, in \cite{kempe2003maximizing, chen2009efficient, chen2010scalable}.
We also conducted experiments on a smaller collaboration network Hep-Th having 7,610 nodes and 31,502 directed edges \cite{newman2001structure}. As the results obtained were very similar, we present the results for only the NetHEPT dataset.
%
%
%
For two-phase diffusion, as a na\"ive first approach as mentioned earlier, we consider equal budget split, $k_1=k_2=\frac{k}{2}$, and $d=D$. 

\begin{remark}
In the two-phase influence maximization problem, seed selection is not computationally intensive, but seed evaluation is. At the end of the first phase, only one $Y$ is possible in practice; however, for the purpose of evaluation as part of simulations, we need to consider $\mathcal{M}_1$ (Monte-Carlo iterations for first phase) number of $Y$'s. This severely restricts the size of network under study. We believe the NetHEPT dataset suffices for our study owing to its social networks-like features and its wide usage for experimentation in the literature. 
\end{remark}

\begin{remark}
The simulations can also be run using a single level of Monte-Carlo iterations instead of two levels as described above.
For instance, instead of deciding the diffusion over each edge dynamically as in the above approach, one can decide an entire live graph in advance so that there is no requirement of separate Monte-Carlo iterations for the two phases.
However, the number of Monte-Carlo iterations (live graphs) required to compute the value with same variance as the above approach would be $\Theta(\mathcal{M}_1 \mathcal{M}_2)$.
\end{remark}

We now list the parameter values for the considered algorithms, specifically for the LM dataset.
For the detailed FACE algorithm, the reader is referred to \cite{de2005tutorial}.
We initialize the method with distribution $(\frac{\gamma}{n},\ldots,\frac{\gamma}{n})$, that is, each node has a probability of $\frac{\gamma}{n}$ of getting selected in any sample set in the first iteration (where $\gamma$ is the budget which is $k,k_1,k_2$ for single phase diffusion, first phase, and second phase, respectively).
In any iteration, the number of samples (satisfying budget constraint) is bounded by
 $\mathcal{N}_{\text{min}}=n$ and
$\mathcal{N}_{\text{max}}=20n$, the number of elite samples (samples that are deemed to have good enough function value) is
$\mathcal{N}_{\text{elite}}=\lceil \frac{n}{4} \rceil$.
%
We use a weighted update rule for the distribution where, in any given iteration, the weight of any elite sample is proportional to its function value. 
The smoothing factor (telling how much weight is to be given to the current iteration as against the previous iterations) that we consider is $\alpha = 0.6$.
%
In our simulations, we observed that in most cases, the FACE algorithm converged in 5 iterations (extending till 7 at times) by giving a reliable solution ({\em reliable} refers to the case wherein the method deduces that it has successfully solved the problem). Also, the total number of samples drawn in any iteration was $n$ in almost all cases (it did not exceed $2n$ in any iteration).
That is, the total number of samples over all iterations was approximately $5n$.
So for direct comparison with SPIC and RMax, we consider $5n$ permutations in order to compute the approximate Shapley values of all the nodes~\cite{narayanam2010shapley},
and $5n$ sampled sets for RMax.

\setlength{\tabcolsep}{.5em}
\begin{table}[t]
\caption{\mbox{Gain of two-phase diffusion 
over single phase} one on LM dataset (WC model) with ${k = 6, k_1 = k_2 = 3, d = D}$}
\centering
\begin{footnotesize}
\hspace{-2mm}
\begin{tabular}{c|c|c|c|c|c|c}
\hline \hline
\T \B
\multirow{4}{*}{\hspace{-4mm} Method} & \multicolumn{3}{c|}{Expected spread} &  \multicolumn{3}{c}{Running Time for} \\
\T \B
	& \multicolumn{3}{c|}{} & \multicolumn{3}{c}{seed selection (seconds)} \\
\cline{2-7}
\T \B
	& Single & Two- & \%  & Single & Myopic & Farsight  \\
\T \B
					& phase & phase & gain & phase  & 2-phase & 2-phase \\
\hline
\T \B
\hspace{-4mm} FACE 	& 46.2 & 50.7 &  9.7 & 15 & 29  & 1209 \\
\T \B
\hspace{-4mm} SPIC 	& 45.9 & 50.4 &  9.8 & 16 & 31   & 1272 \\
\T \B
\hspace{-4mm} Greedy & 46.2 & 49.7 &  7.6 & 10  &  11 &  390  \\
\T \B
\hspace{-4mm} PMIA 	& 46.2 & 49.4 &  6.9 & 0.2 & 0.2 & 0.2  \\
\T \B
\hspace{-4mm} GDD 	& 45.8 & 49.3 &  7.6 & 0.002 & 0.002 & 0.002 \\
\T \B
\hspace{-4mm} WD 	& 45.7 & 48.7 &  6.6 & 0.002 & 0.002 & 0.002 \\
\T \B
\hspace{-4mm} SD	& 40.5 & 44.5 &  9.9 & 0.002 & 0.002 & 0.002  \\
\T \B
\hspace{-4mm} RMax 	& 35.9 & 46.6 &  29.8 & 6 & 12  & 751 \\
\hline \hline
\end{tabular}
\end{footnotesize}
\label{tab:percent_improve}
\vspace{-2mm}
\end{table}
\setlength{\tabcolsep}{.6em}

\subsection{Simulation Results}

Throughout the rest of this paper, we present results for a few representative settings. We have conducted simulations over a large number of settings and the results presented here are very general in nature.

  \begin{figure*}
 \begin{minipage}{5.5cm}
   \includegraphics[scale=.42]{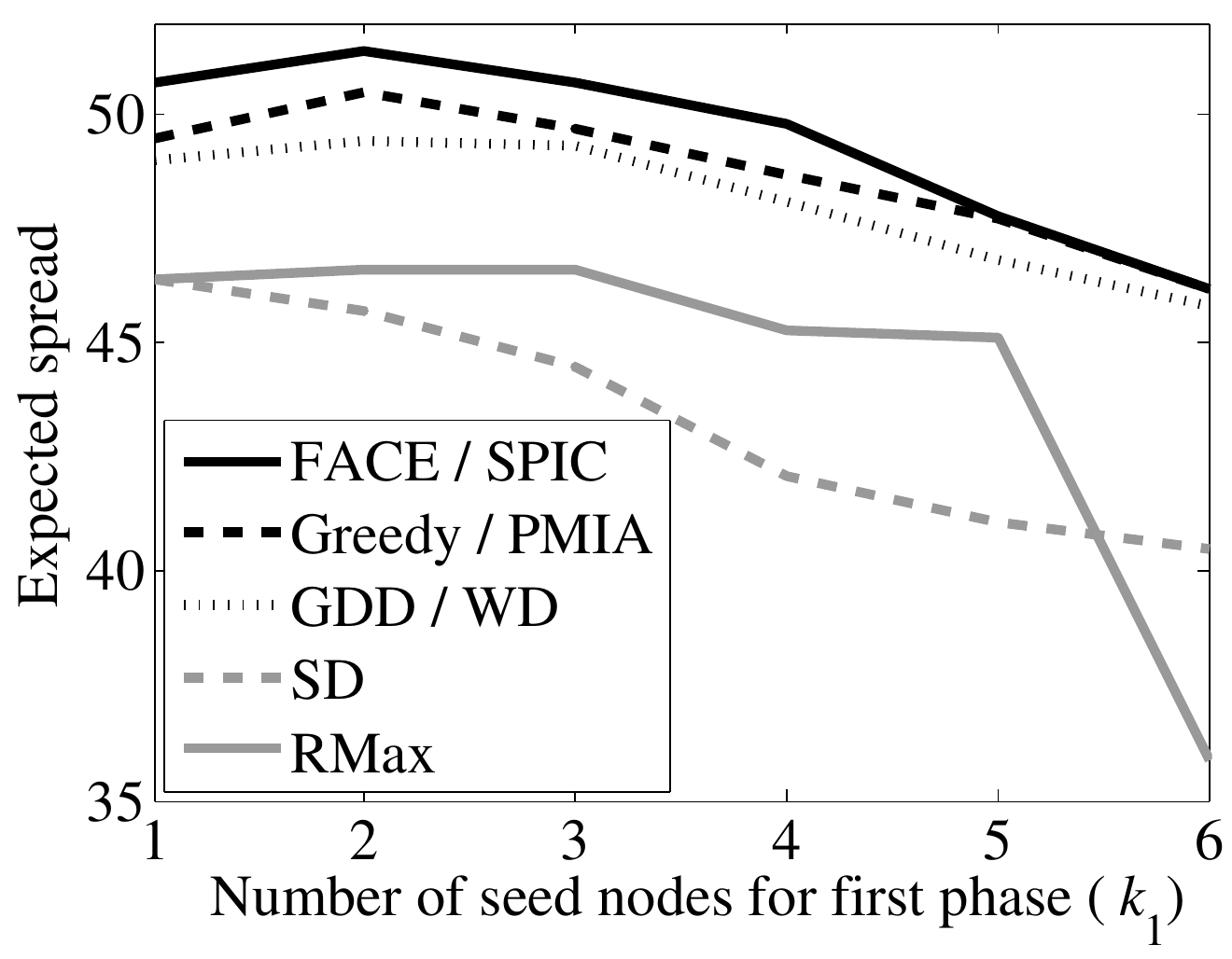} 
 \\ \centering (a)
\end{minipage}
\begin{minipage}{5.5cm}
\includegraphics[scale=.415]{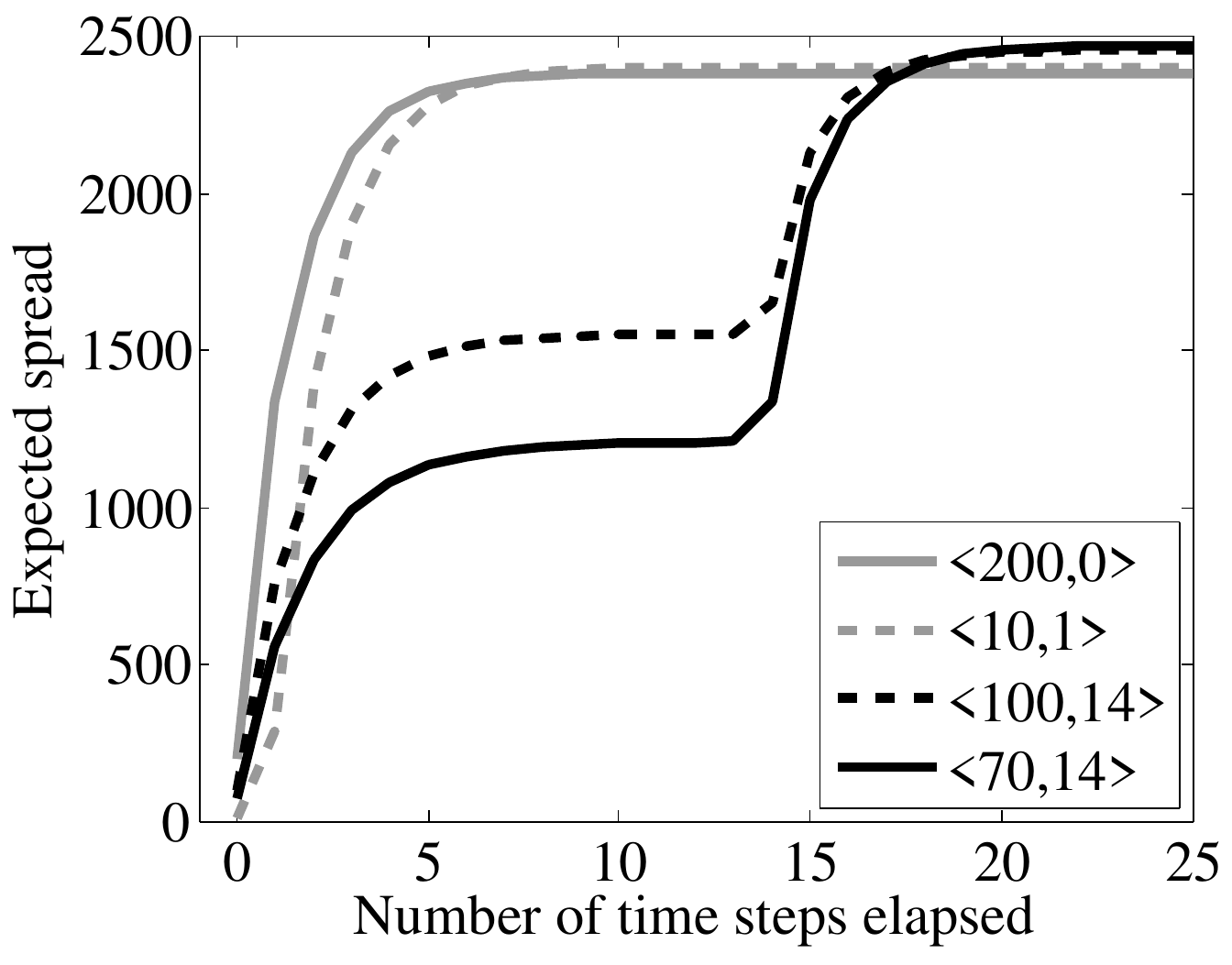} 
  \\ \centering (b)
\end{minipage}
 \begin{minipage}{5.75cm}
 \includegraphics[scale=.415]{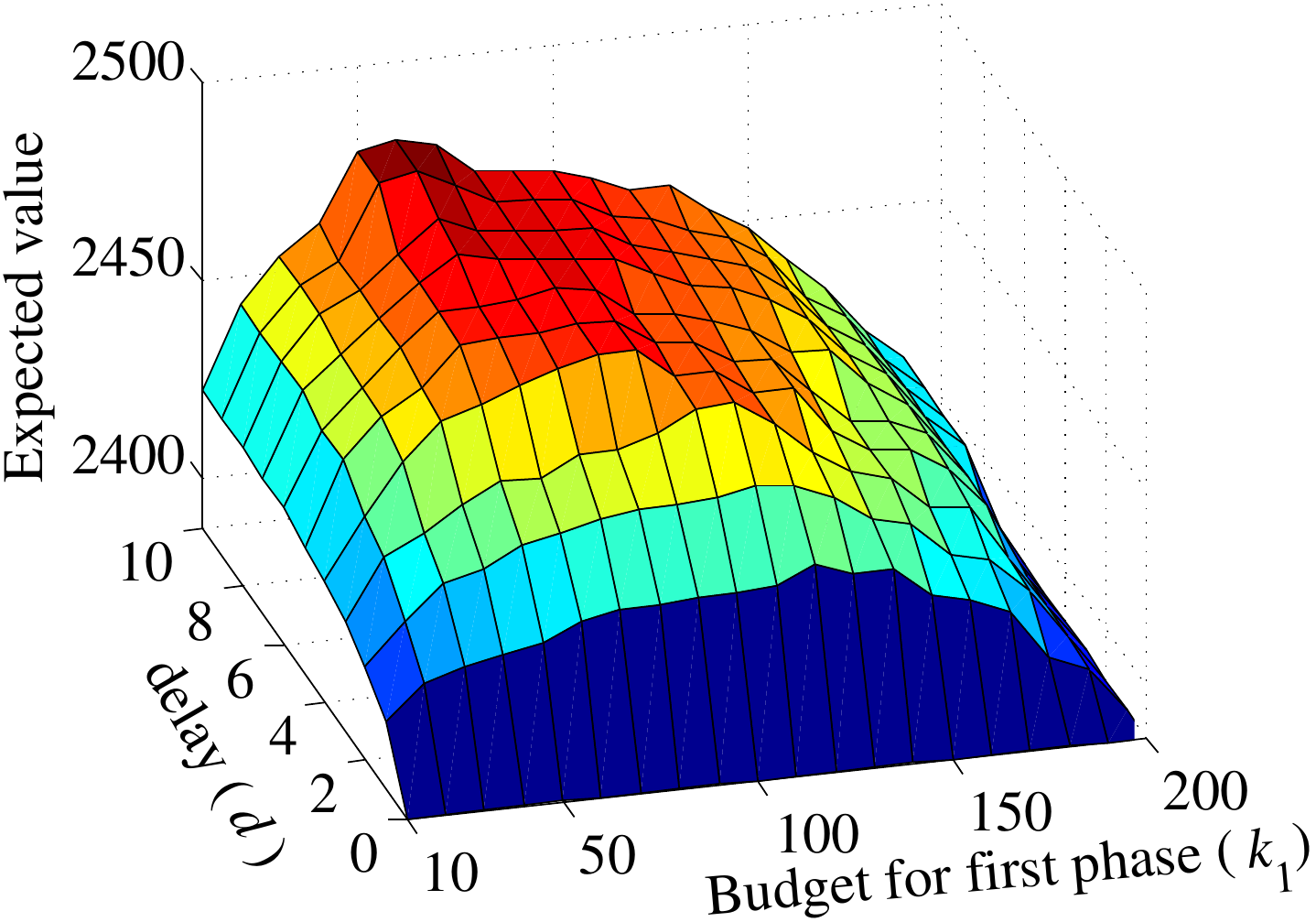} 
     \\ \centering (c)
\end{minipage}
   \caption{ 
   (a) Performance of algorithms for different values of ${k_1}$ on LM dataset under WC model (${k=6}$, ${d=D}$),
(b) Typical progression of diffusion with time for different $<$$k_1,d$$>$ pairs on NetHEPT  under WC model ($k=200$),
(c) 3D plot considering a range of $<$$k_1,d$$>$ pairs for $\delta=1$ on NetHEPT  under WC model ($k=200$)
   }
   \label{fig:plots_mpid}
  \end{figure*}

\begin{observation}
FACE algorithm is very effective for single phase influence maximization, performing at par with greedy and PMIA or better for most values of $k$. SPIC also performs almost at par with them. To justify the effectiveness of two-phase diffusion process, it was necessary to consider these high performing single phase algorithms. Furthermore, GDD heuristic performs very closely to these algorithms, while taking orders of magnitude less time.
\end{observation}

Table~\ref{tab:percent_improve} shows the improvement of the na\"ive two-phase diffusion over single phase one for the considered algorithms on the LM dataset (WC model).
The performances of myopic and their farsighted counterparts were observed to be almost same (the maximum difference in the expected spread was observed to be 0.2 on the scale of 77 nodes),
so they share a common column for the expected spread.
These results, in conjunction with other results for $k_1\neq k_2$ and $d<D$ (which are not presented here), show that the myopic algorithms perform at par with the farsighted ones, while running orders of magnitude faster.
A possible reason for the excellent performance of myopic algorithms is that, the first set of $k_1$ nodes selected by most influence maximization algorithms, are generally the ones which would give a large enough observation and a well refined search space for the second phase seed set. 
Also, as mentioned earlier, there is no distinction between myopic and farsighted algorithms for heuristics such as PMIA, GDD, WD, SD, that do not consider the actual objective function for seed selection, so their running times also are the same.
The results obtained using TV model were qualitatively similar with a very slight dip in the \% gain with respect to the expected spread; the running times were significantly lower for most algorithms owing to lower edge probabilities in TV model as compared to WC model (in the case of LM dataset) and so the diffusion/simulation would terminate faster.

The results for NetHEPT are presented in Table~\ref{tab:impwithk}. For the purpose of this section, we need to only look at the first rows ($k_1=k_2$) of both WC and TV models.

\begin{observation}
Though it is clear that two-phase diffusion strictly performs better than single phase diffusion, the amount of improvement depends on the value of $k$ as well as the diffusion model under consideration (see Table~\ref{tab:impwithk}). 
\end{observation}

Note that the amount of improvement is significant, especially when the company is concerned with monetary profits or a long-term customer base.
We now attempt to further improve what we can get by using the two-phase diffusion.

\begin{table}[t!]
\caption{
\% Improvement of two-phase diffusion over single phase depending on $k$}
\centering
\begin{tabular}{c|c|c|c|c|c}
\hline \hline
\T \B
Model	&	$k \rightarrow$				&	50		&	100		&	200		&	300
\\ \hline \T \B
\multirow{3}{*}{WC}	&	\% Improvement ($k_1=k_2$) &	3.5		&	1.8		&	3.5		&	4.4
\\ \T \B
&	Opt. \% improvement &	4.5		&	2.0		&	4.0		&	4.5
\\ \T \B
	&	Optimizing $k_1$				&	15		&	35		&	70		&	105
\\ \hline \T \B
\multirow{3}{*}{TV}	&	\% improvement ($k_1=k_2$)	&	5.0		&	5.4		&	5.4		&	4.8
\\ \T \B
&	Opt. \% improvement	&	6.0		&	6.0		&	6.0		&	5.0
\\ \T \B
	&	Optimizing $k_1$						&	18		&	35		&	70	&	105
\\ \hline \hline
\end{tabular}
\label{tab:impwithk}
\vspace{-2mm}
\end{table}

\section{Getting the Best out of Two-Phase Diffusion}
\label{sec:practical}

Till now, we assumed $k_1=k_2$ and $d=D$, and we needed to determine the best seed sets (a) of size $k_1$ for first phase and (b) of size $k_2$ for second phase based on the observed diffusion after a delay of $d$ time steps. However, in practical situations, there is also a need to determine (c) an appropriate split of the total budget $k$ into $k_1$ and $k_2$ as well as (d) an appropriate delay $d$ (we have proved that $d=D$ is optimal in absence of temporal constraints, but this may not be the case in their presence). In this section, we address these issues.
Henceforth, we use only farsighted algorithms, as they take the values of $k_2$ and $d$ into account while computing the objective function value.

  \begin{figure*}
 \begin{minipage}{5.6cm}
   \includegraphics[scale=.415]{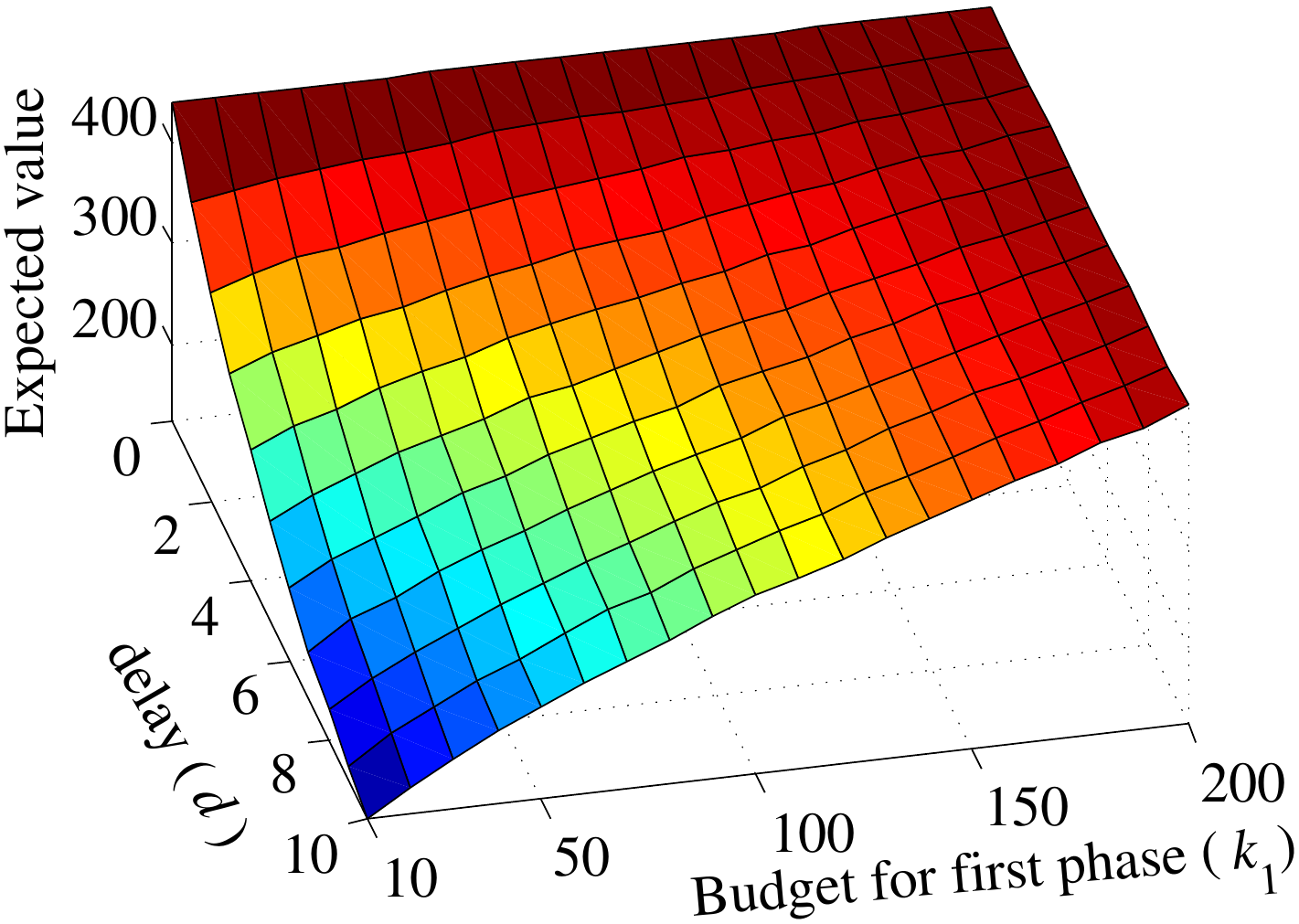} 
 \\ \centering (a) $\delta=0.85$
\end{minipage}
\begin{minipage}{5.6cm}
\includegraphics[scale=.415]{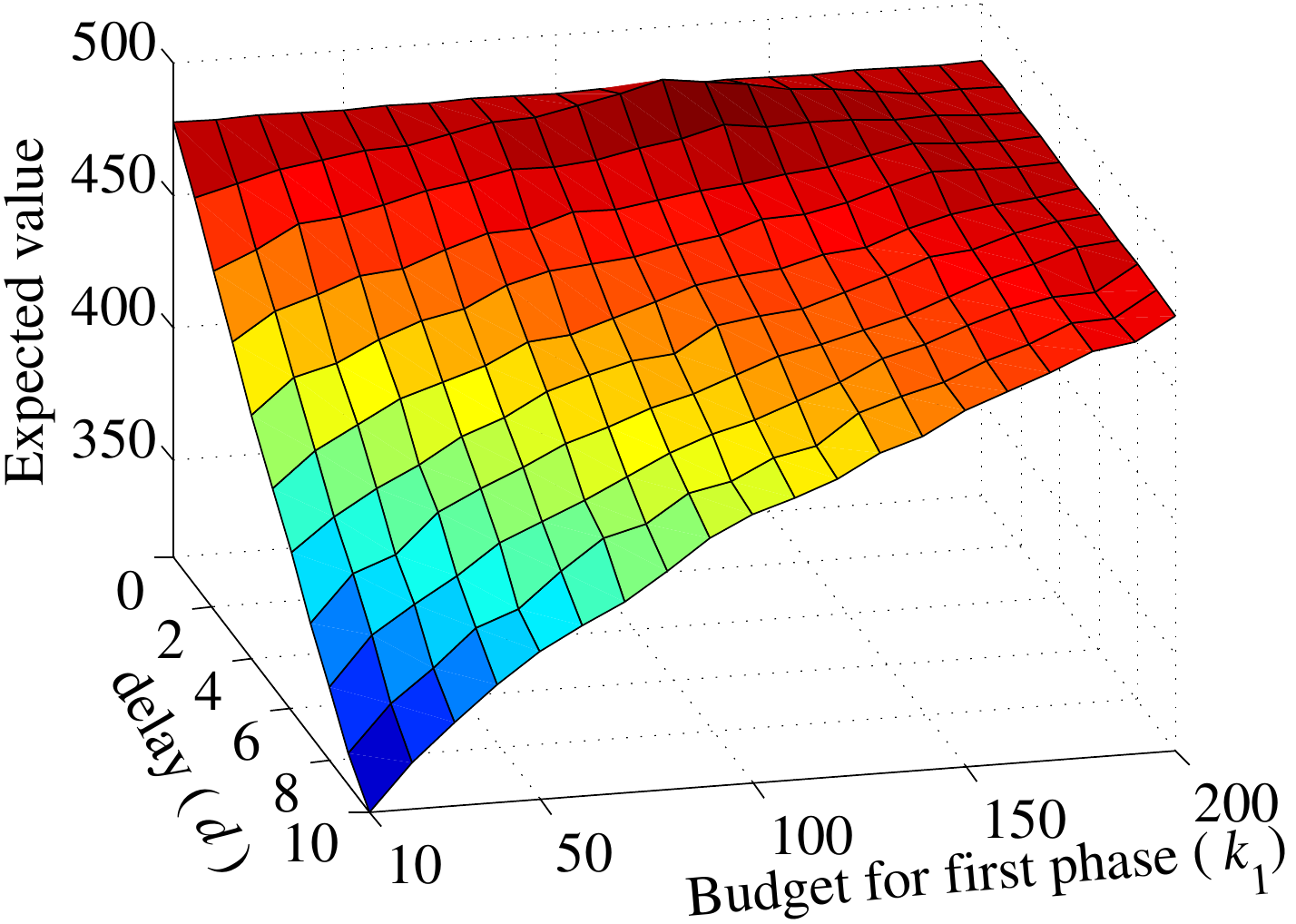}
  \\ \centering (b) $\delta=0.95$
\end{minipage}
 \begin{minipage}{5.6cm}
\includegraphics[scale=.415]{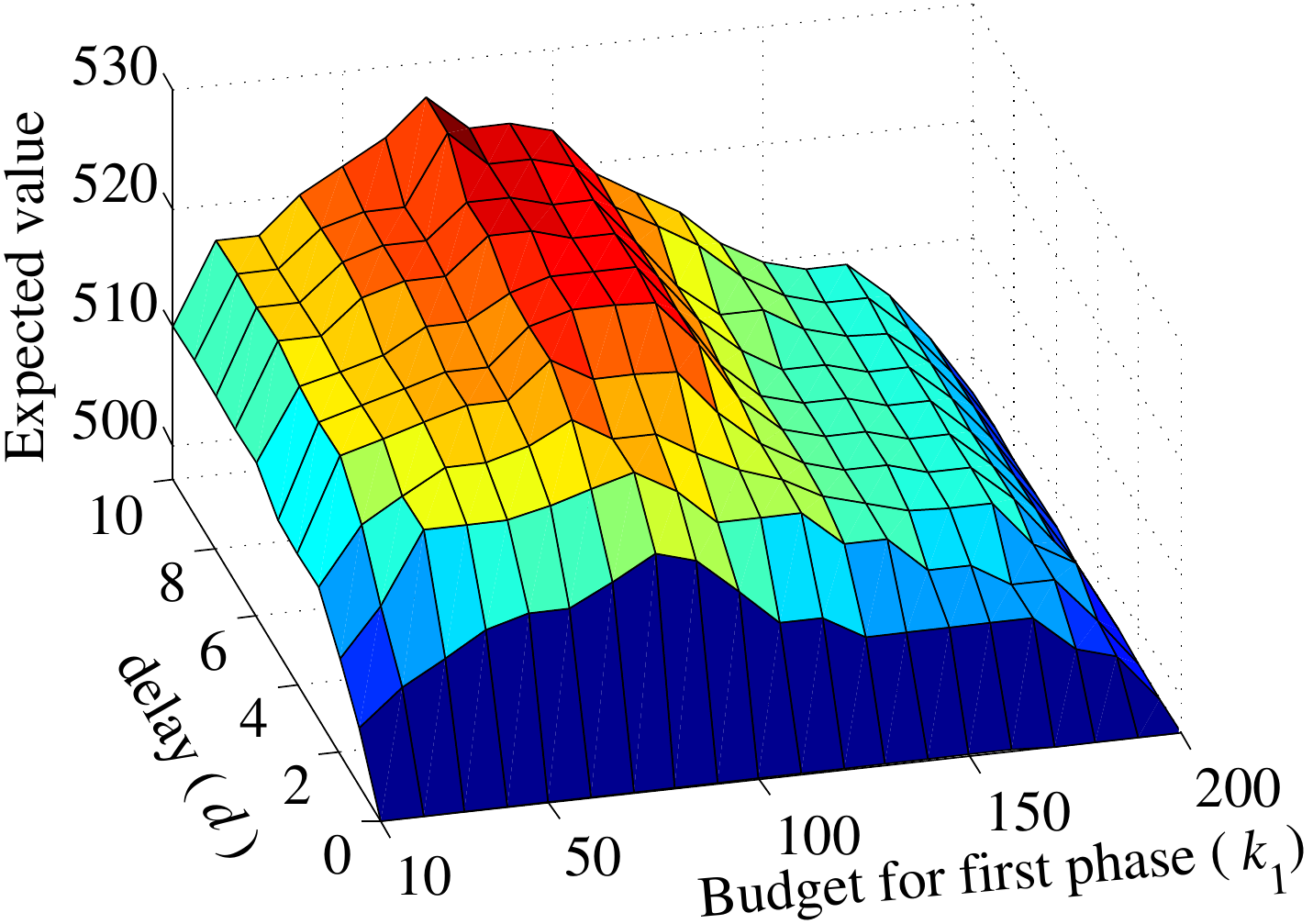}
     \\ \centering (c) $\delta=1$
\end{minipage}
   \caption{ 
   3D plots considering a wide range of $<$$k_1,d$$>$ pairs for different values of $\delta$ on NetHEPT dataset under TV model ($k=200$)
   (note the reversed delay axis in (c) as compared to (a-b))
   }
   \label{fig:3Dplots_TV}
     \vspace{-.2cm}
  \end{figure*}

\subsection{Budget Splitting}

Here we address the problem of splitting $k$ optimally between the two phases, that is, determining an optimal $k_1$ and hence $k_2$.
Note that when $k_2$ is not fixed, the objective function $\mathbb{F}(S_1,d,k_2)$ is no longer monotone with respect to the first phase seed set $S_1$.
For instance, $\mathbb{F}(\{\},d,k-|\{\}|) = \mathbb{F}(\{\},d,k) = \sigma(S^O)$ where $S^O$ is the optimal seed set for single phase,
while for any $|S^\#|=k$, $\mathbb{F}(S^\#,d,k-|S^\#|) = \mathbb{F}(S^\#,d,0) = \sigma(S^\#)$. Unless $S^\#$ is an optimal seed set for single phase, we will have $ \sigma(S^O)> \sigma(S^\#)$ and hence $\mathbb{F}(\{\},d,k-|\{\}|)>\mathbb{F}(S^\#,d,k-|S^\#|)$, even though $\{\} \subset S^\#$.


Using FACE, we can implicitly optimize over $k_1$ and $S_1$ (such that $|S_1|=k_1$) simultaneously by allowing each data sample to consist of a value of $k_1$ sampled from $\{1,\ldots,k\}$, as well as a sampled set $S_1$ of size $k_1$. 
\begin{remark}
For faster convergence, in the first iteration, instead of choosing each node $i$ in the set with probability $\frac{k_1}{n}$, we choose it with probability $q_i = \frac{k_1 w_i}{\sum_i w_i}$, where $w_i$ is as in Equation~(\ref{eqn:gdd_value}).
In cases wherein the value of $q_i$ exceeds 1, we distribute the surplus to other nodes with values less than 1, in proportion of their current values. We repeat this until all nodes have values at most 1. This process of distributing the surplus value is to ensure that 
the expected size of the sampled set does not drop below $k_1$.
%
The rest of the iterations follow as per the standard FACE algorithm. 
\end{remark}
%
In RMax method, for every sample, $k_1$ is chosen u.a.r. (uniformly at random) from $\{1,\ldots,k\}$ and hence a set $S_1$ of cardinality $k_1$ is sampled. The output set is one that maximizes the objective function among the sampled sets.
As there is no implicit way to optimize over $k_1$ in rest of the algorithms,
we do the following: as we add nodes one by one to construct the set $S_1$, we keep track of the maximum value attained so far, to determine a value maximizing set $S_1$ of size $k_1 \leq k$. 
%

Figure~\ref{fig:plots_mpid}(a) presents the results of different budget splits for the considered algorithms on LM dataset (results are for WC model, results for TV model were qualitatively similar).
We also studied various budget splits for NetHEPT dataset using both WC and TV models, the results of which are provided in
Figures~\ref{fig:plots_mpid}(c) and \ref{fig:3Dplots_TV}(c)
(see $d=10$; we have limited $d$ to $10$ for the purpose of clarity; the observations for $d>10$ were almost same as that for $d=10$)
and also Figure~\ref{fig:plot_all_deltas_hep_tv} ($\delta=1.00$).
The results for different values of $k$ are provided in
Table~\ref{tab:impwithk}.
These results show that our na\"ive first guess of splitting the budget equally was a good one, even though other splits give marginally higher values
(considering $d=D$).

\begin{observation}
For the datasets considered, under all settings (different diffusion models and values of $k$), a split of $k_1:k_2 \approx \frac{1}{3}:\frac{2}{3}$ is observed to be optimal. 
\end{observation}

A possible reason for $k_1 \approx k_2$ being a good guess is
a trade-off between (a) the size of the observed diffusion and (b) the exploitation based on the observed diffusion. If the value of $k_1$ is too low, not many nodes may be influenced and so we may not be able to observe the diffusion to a considerable extent, leaving us with little information for selecting the seed nodes for the second phase. On the other hand, if the value of $k_2$ is too low, we may not be able to select enough number of seed nodes for the second phase to exploit the information obtained from the observed diffusion.
The optimal split $k_1:k_2 \approx \frac{1}{3}:\frac{2}{3}$ (a skew towards lower values of $k_1$) can perhaps be attributed to the fact that the first set of seed nodes selected by most algorithms, are very influential, and it is not necessary to allocate half of the budget to first phase in order to obtain a large enough observable diffusion.

  \begin{figure}
  \centering
  \includegraphics[scale=0.42]{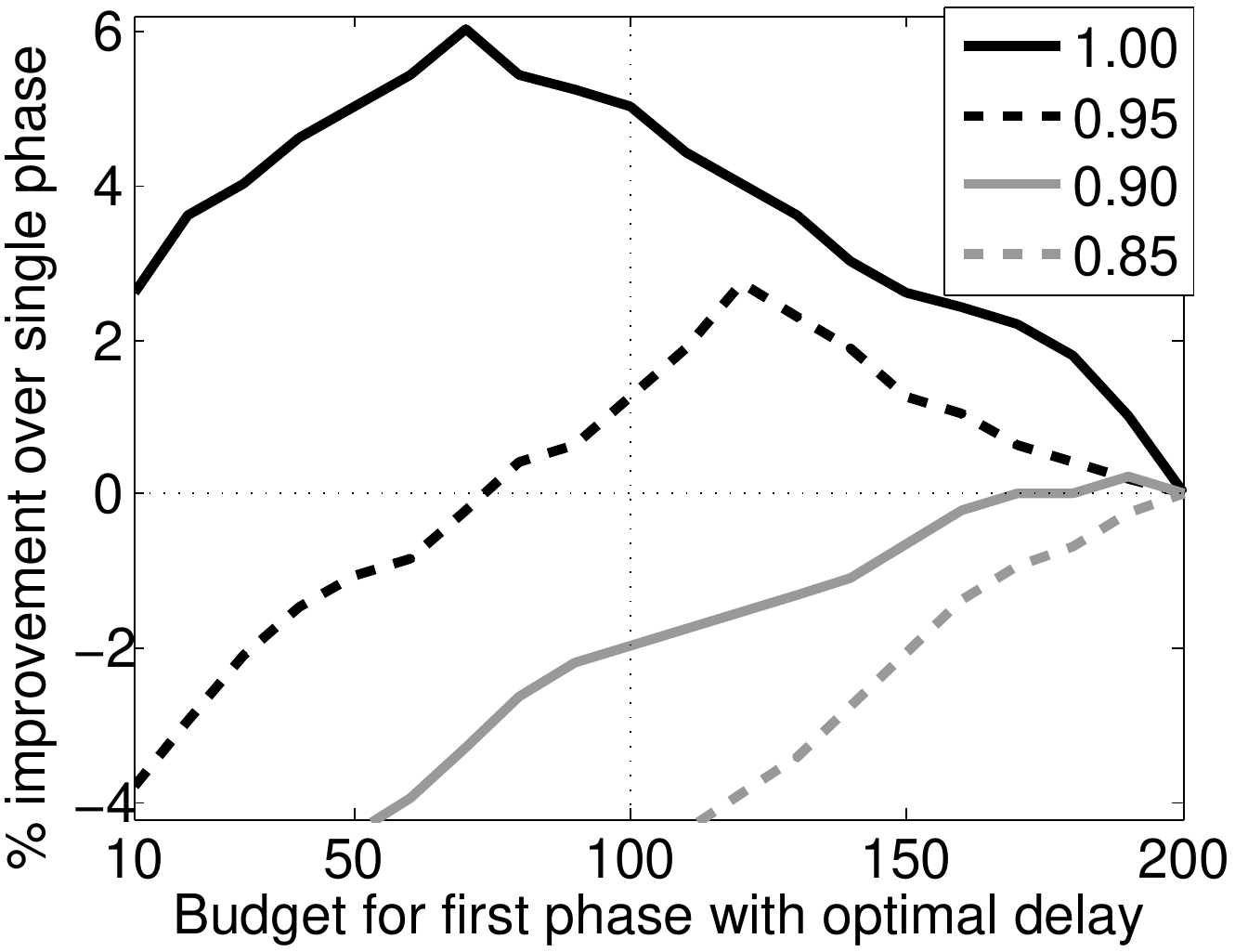}
  \caption{Typical results of splitting ${k=200}$ with optimal $d$ $(\geq 1)$ for different ${\delta}$'s ($d=D$ for $\delta=1.00$, $d=1$ for other $\delta$'s) on NetHEPT  (TV model) (2D views of plots in Figure~\ref{fig:3Dplots_TV} corresponding to the optimal $<$$k_1,d$$>$ pairs)}
  \label{fig:plot_all_deltas_hep_tv}
  \vspace{-2mm}
  \end{figure}
  
  \subsection{Scheduling the Second Phase}
  \label{sec:temporal}
  
  It is clear that a two-phase diffusion would result in a higher influence spread than the single-phase one. This brings us to address the following questions: (a) why not use two-phase diffusion all the time? and (b) why not wait for the first phase to complete its diffusion process before starting the second phase? It is to be noted that the standard IC model fails to capture the effects of time taken for the diffusion process. 
  A more realistic objective function would capture not only the influence spread, but also its rate.
  One such objective function could be 
  $
  \nu(S) = \sum_{t=0}^\infty \Gamma(t) \sigma_{(t)}(S)
  $,
  where $\Gamma(\cdot)$ is a non-increasing function such that $\Gamma(t) \leq 1$ for all values of $t$, and $\sigma_{(t)}(S)$ is the expected number of newly activated nodes at time step $t$.
  
  Alternatively, let $t_j^{X,S}$ be the minimum number of time steps in which node $j$ can be reached from set $S$ in live graph $X$.
  Then $\Gamma(t_j^{X,S})$ is the value obtained for influencing node $j$ in live graph $X$, and $\sum_X p(X) \Gamma(t_j^{X,S})$ is the expected value obtained for influencing node $j$.
  So the expected influence value obtained starting from a seed set $S$ is $\nu(S) = \sum_j  \sum_X p(X) \Gamma(t_j^{X,S})$.
  Note that if $\Gamma(t) = 1$ for all $t$, then $\nu(\cdot)$ reduces to $\sigma(\cdot)$.
  Thus we modify our two-phase objective function 
  by incorporating $\Gamma(t)$.
  %
  %
  \begin{theorem}
  $\nu(\cdot)$  is non-negative, monotone increasing, and submodular, for any non-increasing function $\Gamma(\cdot)$ where $0 \leq \Gamma(t) \leq 1, \; \forall t$. 
  \end{theorem}
  \begin{proof}
  The non-negativity of $\nu(\cdot)$ is direct from the non-negativity of $\sigma_{(t)}(\cdot)$.
  Now, it is clear that $t_j^{X,S} \geq t_j^{X,T}$ for any $S \subset T$, and owing to $\Gamma(\cdot)$ being a non-increasing function, we have $\Gamma(t_j^{X,S}) \leq \Gamma(t_j^{X,T})$.
  Since this is true for any live graph $X$, we have $\sum_X p(X) \Gamma(t_j^{X,S}) \leq \sum_X p(X) \Gamma(t_j^{X,T})$. Also, since this is true for any node $j$, we have $\sum_j  \sum_X p(X) \Gamma(t_j^{X,S}) \leq \sum_j  \sum_X p(X) \Gamma(t_j^{X,T})$ or equivalently, $\nu(S) \leq \nu(T)$. 
  This proves the monotone increasing property of $\nu(\cdot)$.
  
  For the purpose of proving its submodularity, let us define another function $\psi_j^X(S) = \Gamma(t_j^{X,S})$. So $\nu(S) = \sum_j  \sum_X p(X) \psi_j^X(S)$, that is, $\nu(\cdot)$ is a non-negative linear combination of the functions $\psi_j^X(\cdot)$.
  Consider arbitrary sets $S$ and $T$ and an arbitrary node $i$ such that $S \subset T$ and $i \in N\setminus T$.
  We first prove the submodularity of $\psi_j^X(\cdot)$ for an arbitrary node $j$ and a live graph $X$ using two possible cases. 
  In the first case, if addition of $i$ to the set $T$ does not reduce the number of time steps required to reach node $j$, then $\psi_j^X({T\cup\{i\}}) = \psi_j^X(T)$.
  In the second case,  if addition of $i$ to the set $T$ reduces the number of time steps required to reach the node, then $\psi_j^X({T\cup\{i\}}) = \psi_j^X({\{i\}}) = \psi_j^X({S\cup\{i\}})$.
  In both the cases, $\psi_j^X({S\cup\{i\}}) - \psi_j^X(S) \geq \psi_j^X({T\cup\{i\}}) - \psi_j^X(T)$. 
  This proves the submodularity of $\psi_j^X(\cdot)$ and hence of their non-negative linear combination $\nu(\cdot)$.
  \end{proof}
  Thus following argument similar to that in Section~\ref{sec:compute}, the two-phase objective function (taking time into consideration) can be well approximated using greedy algorithm 
  for seed selection in the second phase; and GDD heuristic can be used as an effective proxy for greedy.
  %
  Note that GDD heuristic would perform very well for the temporal objective function $\nu(\cdot)$ because it maximizes the number of nodes influenced in the immediately following time step. It would be an excellent algorithm when $\Gamma(1)$ is significantly larger than $\Gamma(t)$ for $t \geq 2$.
  In our simulations, we consider $\Gamma(t) = \delta^t$ where $\delta \in [0,1]$ (this is generally the first guess for a decay function in several problems).

\begin{table}[t]
\caption{Performance of  FACE with implicit optimization versus that with exhaustive search}
\centering
%
\begin{tabular}{c||c||c|c|c||c|c|c}
\hline \hline
\T \B
\multirow{3}{*}{$\delta$} & Single & \multicolumn{6}{c}{Two-phase with optimal $<$$k_1,d$$>$} 
\\
\cline{3-8}
& \multicolumn{1}{c||}{phase} &  \multicolumn{3}{c||}{\T \B Implicit opt.} & \multicolumn{3}{c}{\T \B Exhaustive}
\\ \cline{3-8}
\T \B
& value & $k_1$ & $d$ & value & $k_1$ & $d$ & value 
\\ \hline
\T \B
0.75 & 33.2	&	6	&	0	& 33.2	& 6 & 0  	& 33.2	
\\ \T \B
0.80 & 36.0 &	6	&	0	& 36.0	& 5 & 1  	& 36.7	
\\ \T \B
0.85 & 38.0	& 	5	&	2	&  39.1 & 5 & 1  	& 39.7	
\\ \T \B
0.90 & 40.6 & 	5	&	2	&  41.8	& 4 & 1 	& 42.4	
\\ \T \B
0.95 & 42.9	& 	4	&	1	&  46.1	& 4 & 2 	& 46.5 	
\\ \T \B
1.00 & 46.2	& 	2	&	$D$	&  51.4 & 2 & $D$ 	& 51.4	
\\ \hline \hline
\end{tabular}
%
\label{tab:face_temporal}
\vspace{-2mm}
\end{table}

Now our objective is to not only find an optimal $k_1$, but also an optimal delay $d$.
%
%
We have seen that FACE algorithm implicitly computes influential seed nodes while simultaneously optimizing over $k_1$. Now in addition to a sampled value of $k_1$ and a sampled set of cardinality $k_1$, we allow each data sample to also contain a value of $d$, sampled from $\{1,\ldots,D\}$. Table~\ref{tab:face_temporal} shows that the differences between (a) the spread achieved using this implicit optimization method and (b) that achieved using exhaustive search over $k_1$ and $d$, for different $\delta$'s on LM dataset, are low. 
The time taken for implicit optimization was observed to be approximately $\frac{1}{k D}$ of that taken for exhaustive search.
This shows the effectiveness of FACE algorithm for getting the best out of two-phase diffusion by addressing the combined optimization problem.

As mentioned earlier, for NetHEPT dataset also,
we observed that for $\delta=1$, it is optimal to allocate one-third of the budget to first phase and delay $d=D$. 
For $\delta \leq 0.85$, it is optimal to use single-phase diffusion. 
For intermediate values of $\delta$, it is optimal to allocate most of the budget to the first phase with a delay of one time step; the necessity of allocating most of the budget to the first phase increases as $\delta$ decreases. Figures~\ref{fig:3Dplots_TV}(a-c) and \ref{fig:plot_all_deltas_hep_tv} show this in an elaborate way.

Figure~\ref{fig:plots_mpid}(b) shows how the expected spread progresses with time for different $<$$k_1,d$$>$ pairs on NetHEPT dataset under WC model, given $k=200$.
$<$$200,0$$>$ corresponds to single phase diffusion, $<$$10,1$$>$ corresponds to two-phase diffusion with a random $<$$k_1,d$$>$ pair, $<$$100,14$$>$ corresponds to equal budget split $k_1=k_2=\frac{k}{2}$ with $d=D$, and $<$$70,14$$>$ corresponds to optimal $<$$k_1,d$$>$ pair.
(We have $D=14$ in the plots as the first phase diffusion stagnated after 14 time steps for $k_1=70$ and $100$.) 
These types of plots showing the progression of diffusion with time may help a company to decide the ideal values of $k_1$ and $d$ based on its desired progression.

\subsection{An Efficient Method for the Combined Optimization Problem}
\label{sec:gss}

We have seen that the performance of FACE algorithm is excellent, however, it is computationally intensive and hence impractical for large networks. With this in view, we propose another algorithm that is based on empirical observations in Figures~\ref{fig:plots_mpid}(c) and \ref{fig:3Dplots_TV}(a-c).

We note that the plots are unimodal in nature for the considered representative algorithms and datasets with respect to either $k_1$ (with a good enough interval between consecutive $k_1$'s) or $d$ as variable. We could exploit this nature for maximizing the objective function by using the {\em  golden section search} technique with $k_1$ as the variable, where the objective function itself is computed with an optimal $d$ for that particular $k_1$ (which can be found using golden section search). In the special case of the considered exponential decay function, since the optimal values of $d$ would be very small for almost any $\delta<1$, we find an optimal $d$ for a particular $k_1$ using sequential search starting from $d=0$.
Note that as long as the function does not change its value drastically within small intervals (which would be true for the considered problem), the golden section search technique will give an optimal or near-optimal solution even when the function is not perfectly unimodal, but unimodal when the interval between consecutive $k_1$'s is good enough.

We also explored whether the plots are unimodal with respect to $k_1$ and $d$ simultaneously, so that faster methods such as multidimensional direct search, can be used. 
However, though the plots are observed to be unimodal with respect to $k_1$ and $d$ individually, they are not unimodal with respect to them simultaneously.

\section{Discussion}
\label{sec:conclusion_mpid}
We proposed and motivated two-phase diffusion in social networks, formulated an appropriate objective function, proposed an alternative objective function, developed suitable algorithms for seed selection, and observed their performances using simulations. 
We observed that myopic algorithms perform closely to their farsighted counterparts, while taking orders of magnitude less time.
%
For the combined optimization problem of budget splitting, scheduling, and seed selection, 
we proposed the usage of FACE and golden section search algorithms.
We concluded that: (a) under strict temporal constraints, use single-phase diffusion, (b) under moderate temporal constraints, use two-phase diffusion with a short delay while allocating most of the budget to first phase, (c) in absence of temporal constraints, use two-phase diffusion with a long enough delay with a budget split of $\frac{1}{3}:\frac{2}{3}$ between the two phases (one-third budget for the first phase).
We presented results for a few representative settings; these results are very general in nature.



{\scriptsize{$\bullet$}} 
{\em A Note on the Decay Function\/}:
We considered a very strict decay function (exponential), which resulted in humbling two-phase diffusion for most range of $\delta$. In practice, the decay function would be more lenient, where the value would remain high for first few time steps and then decay at a slow rate. Such a decay function would be more suitable for two-phase diffusion. 
%
Note, however, that our choice of a simple exponential decay function allowed us to draw firm conclusions, which would not have been the case with a sophisticated function. 
One could also account for time by studying the problem in presence of competing diffusions, where a delay may help competitors reach the potential customers first.

{\scriptsize{$\bullet$}} 
{\em Extending to the Linear Threshold (LT) Model\/}:
In this paper, we discussed multi-phase diffusion using IC model, primarily because it is a natural setting for such a diffusion. 
One can as well study multi-phase diffusion using the other most popular model, the LT model \cite{kempe2003maximizing}.
%
In LT model, an influence degree $b_{u,v}$ is associated with every directed edge $(v,u)$
and an influence threshold $\chi_u$ (uniformly distributed in $[0,1]$) is associated with every node $u$. 
The diffusion 
proceeds in discrete time steps
and
a node $u$ is activated when
%
$
\sum_{v\in \mathcal{A}} b_{u,v} \geq \chi_u
$,
where $\mathcal{A}$ is the set of activated nodes.
%
The diffusion stops when no further nodes can be activated.
%
At the beginning of first phase, the thresholds are assumed to be uniformly distributed in $[0,1]$. When the second phase is scheduled to start, we have the information regarding whether a node is active or not.
%
In addition, we have the updated information regarding any inactive node $u$ 
that its threshold is 
greater than $\sum_{v\in \mathcal{A}} b_{u,v}$. So while determining the seed set for second phase, we can exploit this information by assuming its threshold to be uniformly distributed in $\left(\sum_{v\in \mathcal{A}} b_{u,v},1\right]$, instead of a wider (and more uncertain) range of $[0,1]$.

{\scriptsize{$\bullet$}} 
{\em Future work\/}:
This work can be extended to study diffusion 
in more than two phases,
 with respect to the influence spread and the time taken. 
%
%
 We focused on the well-studied IC model and provided a note regarding the LT model; studying multi-phase diffusion under other diffusion models is another direction to look at.
 %
 %
It would be useful to study how multi-phase diffusion can be harnessed to get a desired expected spread with a reduced budget.
%
%
It would be of theoretical interest to prove or disprove if there exists an algorithm that gives a constant factor approximation for the problem of two-phase influence maximization.
%
%
It would also be interesting to study equilibria in a game theoretic setting where multiple campaigns consider the possibility of multi-phase diffusion.

\section*{Acknowledgments}
The original publication appears in IEEE Transactions on Network Science and Engineering, volume 3, number 4, pages 197-210 and is available at http://ieeexplore.ieee.org/abstract/document/7570252/.
A previous, very preliminary, concise version of this paper is published in Proceedings of The 14th International Conference on Autonomous Agents and Multiagent Systems, 2015 \cite{dhamal2015multiphase}.
This work has been partially supported by Adobe Research Labs, Bangalore. At the time when this work was done, Swapnil Dhamal was a recipient of IBM Ph.D. Fellowship and Prabuchandran K. J. was a recipient of TCS Doctoral Fellowship.
The authors thank Surabhi Akotiya, Rohith D. Vallam, and Debmalya Mandal for useful discussions.
 
%
%

\bibliographystyle{IEEEtran}
\bibliography{MPID_TNSE_references}



\end{document}